\documentclass[11pt]{article}
\usepackage[T1]{fontenc}
\usepackage{lmodern}
\usepackage[margin=1in]{geometry}
\usepackage{amsmath,amssymb,amsthm,mathtools,microtype}
\usepackage{booktabs,needspace,float}
\usepackage{xcolor}
\usepackage{tikz}
\usetikzlibrary{arrows.meta,calc,positioning}
\usepackage[colorlinks=true,linkcolor=blue!45!black,citecolor=blue!45!black,
urlcolor=blue!45!black,breaklinks=true]{hyperref}
\hypersetup{pdftitle={Unbounded Holevo additivity gaps in finite dimensions},
pdfauthor={Jinzhao Wang},
pdfsubject={Quantum channels, minimum output entropy, and quantitative strong convergence}}
\numberwithin{equation}{section}
\newtheorem{theorem}{Theorem}[section]
\newtheorem{lemma}[theorem]{Lemma}
\newtheorem{proposition}[theorem]{Proposition}
\newtheorem{corollary}[theorem]{Corollary}
\newtheoremstyle{boldremark}{.5\topsep}{.5\topsep}
  {\normalfont}{}{\bfseries}{.}{5pt plus 1pt minus 1pt}{}
\theoremstyle{boldremark}\newtheorem{remark}[theorem]{Remark}
\newtheorem*{remark*}{Remark}
\DeclareMathOperator{\Tr}{Tr}

\newcommand{\C}{\mathbb C}
\newcommand{\F}{\mathbb F}
\newcommand{\E}{\mathbb E}
\newcommand{\Prob}{\Pr}
\newcommand{\M}{M}
\newcommand{\smin}{S_{\min}}
\newcommand{\norm}[1]{\left\lVert #1\right\rVert}

\newcommand{\ket}[1]{\lvert #1\rangle}
\newcommand{\bra}[1]{\langle #1\rvert}

\title{Unbounded Holevo additivity gaps in finite dimensions}
\author{Jinzhao Wang}
\date{}
\begin{document}
\maketitle
\begin{abstract}
We establish unbounded two-use Holevo additivity gaps in finite dimensions.
For each sufficiently large fixed integer $K$ and all sufficiently large
$n$, we construct channels $T_n$ with output dimension $K^n$, input dimension
$\exp(\Theta_K(n^2))$, and
\[
 \chi(T_n^{\otimes2})-2\chi(T_n)
 \ge n\left[\frac{\log_2K}{K}-2\log_2(1+9/K)\right]-O(1/n).
\]
The gap is linear in output qubits, with an explicit quadratic
input-qubit cost and an explicit threshold on $n$. We also obtain channels
whose single-use Holevo
quantity tends to zero while their two-use Holevo information, and hence
classical capacity, diverges. The channels arise from structured tensor
products of the complementary mixed-unitary channels used in Collins's
free-probabilistic proof~\cite{Collins}. We establish minimum-output-entropy
gaps by combining Collins--Youn's product-group Haagerup inequality~\cite{CY}
with Bordenave--Collins's quantitative strong-convergence estimates~\cite{BC}. We also provide a \href{https://github.com/JWang226/Holevo-Additivity-Gap}{Lean Certificate} for our proofs.
\end{abstract}

\section{Introduction and main results}\label{sec:introduction}

Hastings showed that entangled signal states across channel uses can
increase the classical communication rate by violating additivity of the
Holevo quantity~\cite{Hastings}. The initial gap bounds were small.
Subsequent work sharpened the minimum-output-entropy gap
(equivalently, the Holevo-additivity gap after covariant extension
\cite{Shor})
and dimension bounds \cite{FKM,FK} and extended the range of dimensions
\cite{ASW11,Fukuda}. Belinschi--Collins--Nechita \cite{BCN} obtained gaps
approaching one bit, but the established gaps remained bounded as
channel dimensions grew \cite{HP}.

By contrast, for minimum output R\'enyi entropy of every fixed order
$p>1$, additivity violations growing proportionally to the logarithm of
the output dimension were already known \cite{HW}. These results
motivate the search for large violations of Holevo additivity: can the
gap grow without bound in finite dimensions, and how does it scale with
the input and output numbers of qubits? Physically, this asks whether
entanglement across channel uses can provide an arbitrarily large boost
to the classical communication rate.

A quantum channel is a completely positive
trace-preserving linear map $\Phi:\M_a(\C)\to\M_b(\C)$, where
$\M_d(\C)$ is the algebra of all $d\times d$ complex matrices.%
\footnote{We work on full matrix algebras to describe adjoints on
observables and the action on unnormalized inputs.}
The minimum output entropy
and Holevo quantity are
\begin{align}
 \smin(\Phi)&=\min_\rho S(\Phi(\rho)),\label{eq:def-smin}\\
 \chi(\Phi)&=\sup_{\{p_x,\rho_x\}}
 \left[S\!\left(\sum_xp_x\Phi(\rho_x)\right)
       -\sum_xp_xS(\Phi(\rho_x))\right],\label{eq:def-chi}
\end{align}
where $S(\rho)=-\Tr\rho\log_2\rho$ is the von Neumann entropy.
Entropy concavity allows the minimum in \eqref{eq:def-smin} to be
taken over pure inputs.
The classical coding theorem of Holevo and Schumacher--Westmoreland
\cite{Holevo,SW} identifies the unassisted classical capacity as
\begin{equation}
 C(\Phi)=\lim_{m\to\infty}\frac{\chi(\Phi^{\otimes m})}{m}
        =\sup_{m\ge1}\frac{\chi(\Phi^{\otimes m})}{m}.
 \label{eq:capacity}
\end{equation}
We study the two-use Holevo gap
\begin{equation}
 \Delta_\chi(\Phi)=\chi(\Phi^{\otimes2})-2\chi(\Phi).
 \label{eq:def-gap}
\end{equation}
Since $C(\Phi)-\chi(\Phi)\ge\Delta_\chi(\Phi)/2$, an unbounded two-use gap gives
an unbounded advantage from encoding jointly across channel uses. Additivity of
$C$ between distinct channels remains a separate open problem
\cite[Section~1.3]{LW}, which we do not address here.

Shor's equivalence theorem \cite{Shor} states that additivity of $\chi$
for every pair of finite-dimensional channels is equivalent to
additivity of $\smin$ for every such pair. The quantitative relation
uses his Weyl-covariant extension, which adds a classical input label
controlling unitary conjugations of the output. For channels $\Phi,\Psi$,
their extensions $\widetilde\Phi,\widetilde\Psi$ satisfy
\[
 \chi(\widetilde\Phi\otimes\widetilde\Psi)
 -\chi(\widetilde\Phi)-\chi(\widetilde\Psi)
 =\smin(\Phi)+\smin(\Psi)-\smin(\Phi\otimes\Psi).
\]
The right-hand side is the minimum-output-entropy gap: the entropy
saving obtained by allowing joint inputs. For a channel $\Phi$ and its
complex conjugate $\overline\Phi$, an input switch first combines the
two branches into one channel \cite{FW}. Applying the Weyl extension
then gives a channel $T$ with
\[
 \Delta_\chi(T)\ge2\smin(\Phi)-\smin(\Phi\otimes\overline\Phi).
\]
Section~\ref{sec:conversion} defines these constructions and proves the
relations.

\paragraph{Our approach.}
We compare a lower bound on every individual output entropy with an
upper bound on the joint output entropy for a Bell input.

For one factor, take $K$ unitaries $U_1,\ldots,U_K$ on $\C^N$ and the
joint evolution
\[
 x\longmapsto\frac1{\sqrt K}\sum_{i=1}^K\ket i\otimes U_i x.
\]
Discarding the label gives the mixed-unitary channel
$\rho\mapsto K^{-1}\sum_iU_i\rho U_i^*$; discarding the original
system gives its complementary channel $F$. The $K$-dimensional
output of $F$ records overlaps between the vectors $U_i x$.
For a pure input, the two reduced outputs have the same entropy.

Pair $F$ with its conjugate $\overline F$, which has the same minimum
output entropy. The Bell vector
$\ket{\omega_N}=N^{-1/2}\sum_{j=1}^N\ket{jj}$ satisfies
\[
 (U_i\otimes\overline{U_i})\ket{\omega_N}=\ket{\omega_N}.
\]
The $K$ matching branch pairs combine into one pure state of weight
$1/K$. The Shannon entropy of the resulting mixture bounds the joint
output entropy by $2\log_2K-(\log_2K)/K$.
The corresponding output of $F\otimes\overline F$ has the same
entropy. This estimate holds for every choice of unitaries; the
calculation is given in Lemma~\ref{lem:bell}.

The harder step is to choose unitaries for which \emph{every} output
of $F$ has entropy close to $\log_2K$. We bound output purity using
$S(\sigma)\ge-\log_2\Tr\sigma^2$. The adjoint of $F$ maps each output
observable to a linear combination of $U_i^*U_j$; its norm controls
the expectation on every input. A uniform bound on these combinations
therefore bounds the output purity (Lemma~\ref{lem:purity}).

Following Collins \cite[Sections~2.2--3.3]{Collins}, we first work in
an infinite-dimensional reference model offered by free probability:
unitary shifts on a basis
indexed by reduced words in generators and their inverses. A shift
adds its generator on the left and cancels adjacent inverse pairs.
These shifts form free Haar families, described in
Section~\ref{sec:entropy}. Haagerup's inequality \cite{Haagerup} bounds
their operator expressions in terms of coefficients and word lengths.
For tensor blocks we use Collins--Youn's extension to a product of
free groups, one per tensor factor \cite{CY}.

To make the gap unbounded, use $n$ channels as one block $F_n$.
The entropy bound must control inputs entangled across all $n$
factors.
The adjoint of $F_n$ gives sums of tensor products of $U_i^*U_j$.
In the reference model, these are words of length zero or two in
each group factor. Grouping terms by their nontrivial factors lets
the product-group inequality control the full sum
(Lemma~\ref{lem:cy}).

We transfer this bound to finite matrices using strong convergence,
which controls polynomial operator norms as well as normalized trace
moments \cite{CM,BC}. Norm convergence is essential here: moment
convergence alone can miss exceptional eigenvalues.
For fixed $K,n,\eta$, Bordenave--Collins's tensor-product theorem
\cite[Theorem~9.2]{BC} supplies the required approximation as $N$
grows, using $Kn$ independent Haar unitaries and a fixed finite net
of observables. Section~\ref{sec:finite-entropy} obtains
$\Tr F_n(\rho)^2\le2^\eta K^{-n}(1+9/K)^n$ for every block input,
hence
\[
 \smin(F_n)\ge n\bigl[\log_2K-\log_2(1+9/K)\bigr]-\eta.
\]
A product of $n$ Bell inputs, joining corresponding factors of the
two blocks, gives a joint entropy deficit of at least
$n(\log_2K)/K$. For fixed $\eta$ and sufficiently large fixed $K$,
the minimum-output-entropy gap therefore grows linearly with $n$.
Section~\ref{sec:conversion} applies the input switch and Weyl extension
to obtain the Holevo gap of $T$ in \eqref{eq:channel-T}, proving
Proposition~\ref{prop:qualitative}.

Section~\ref{sec:approximation} gives a quantitative input-dimension bound
by combining the observable tests in one polynomial, expressing the
$K$ generators in each factor as words in two base generators, and
reducing the polynomial to degree one using the algebraic lemmas
in Appendix~\ref{app:operators}. The base generators are
evaluated at independent Haar matrices; the resulting $K$ unitaries
within a factor may be dependent. Appendix~\ref{app:finite-threshold}
establishes the explicit two-generator Haar estimate
(Lemma~\ref{lem:haar}) needed for the dimension bound.
This gives $N=\lceil e^{b_Kn}\rceil$ with $b_K=O(\ln K)$ and entropy
error $O(1/n)$ for the whole block.
For fixed $K$, the final channel requires $\Theta_K(n^2)$ input
qubits and $n\log_2K$ output qubits.

\paragraph{Results.}
Throughout, $\log=\log_2$ and $\ln=\log_e$. The notation
$q_{\rm in}=\log\dim A$ and $q_{\rm out}=\log\dim B$ denotes input and
output qubit counts; rounding up does not affect the asymptotic statements.
For $K\ge2$, let
\begin{equation}
 a_K=\log(1+9/K),\qquad
 \delta_K=\frac{\log K}{K}-2a_K.
 \label{eq:delta}
\end{equation}
To obtain a positive linear gap, we will choose $K$ with $\delta_K>0$.
The inequality $\ln(1+x)\le x$ gives
$\delta_K\ge(\ln K-18)/(K\ln2)$, so $\ln K>18$ is sufficient
for $\delta_K>0$.

\Needspace{18\baselineskip}
\begin{theorem}\label{thm:main}
Fix an integer $K\ge2$, and let
\begin{equation}
 b_K=40(7\ln K+2),\qquad
 n_0(K)=\left\lceil256(1+\ln K)^2\right\rceil.
 \label{eq:dimension-constant}
\end{equation}
For every integer $n\ge n_0(K)$, let $N=\lceil e^{b_Kn}\rceil$.
There is a channel
\[
 T_n:\M_{2N^nK^{2n}}(\C)\to\M_{K^n}(\C)
\]
such that
\begin{equation}
 \chi(T_n)\le na_K+2\log\frac{n+1}{n-1},
 \qquad \chi(T_n^{\otimes2})\ge n\frac{\log K}{K}.
 \label{eq:main-holevo}
\end{equation}
Consequently,
\begin{equation}
 \Delta_\chi(T_n)\ge n\delta_K-4\log\frac{n+1}{n-1}.
 \label{eq:main-gap}
\end{equation}
\end{theorem}

\begin{remark*}[Scaling of the gap]
For each fixed $K$ with $\delta_K>0$, the channels in
Theorem~\ref{thm:main} satisfy, as $n\to\infty$,
\begin{equation}
 \Delta_\chi(T_n)=\Theta_K(n)=\Theta_K(\sqrt{q_{{\rm in},n}}),
 \qquad q_{{\rm out},n}=n\log K.
 \label{eq:scaling}
\end{equation}
Indeed, $4\log\frac{n+1}{n-1}=O(1/n)$, so \eqref{eq:main-gap}
gives a linear lower bound in $n$, while
$\Delta_\chi(T_n)\le2\log(K^n)=2n\log K$ gives the matching upper
bound. The input dimension gives
\begin{equation}
 q_{{\rm in},n}
 =\frac{b_K}{\ln2}n^2+2n\log K+1+o(1),
 \label{eq:input-size}
\end{equation}
which yields the scaling with input qubits in \eqref{eq:scaling}.
The theorem also guarantees
\begin{equation}
 \liminf_{n\to\infty}\frac{\Delta_\chi(T_n)}{q_{{\rm out},n}}
 \ge r_K:=\frac{\delta_K}{\log K}
 =\frac1K-\frac{2\ln(1+9/K)}{\ln K}>0.
 \label{eq:guaranteed-fraction}
\end{equation}
As $K\to\infty$,
$r_K=(1-18/\ln K)/K+O(1/(K^2\ln K))$.\footnote{For example,
$K=\lceil e^{19}\rceil=178\,482\,301$ gives
$\delta_K\approx8.08\times10^{-9}$,
$r_K\approx2.95\times10^{-10}$, and
$n_0(K)=102\,401$.
These values describe the proved sufficient guarantee; the actual
gaps may be larger.}
\end{remark*}

A gap growing linearly with the number of uses cannot be obtained merely by
regrouping uses of one fixed channel. For any fixed $\Phi$, \eqref{eq:capacity}
gives
\[
 \frac{\Delta_\chi(\Phi^{\otimes m})}{m}
 =2\left[\frac{\chi(\Phi^{\otimes2m})}{2m}
         -\frac{\chi(\Phi^{\otimes m})}{m}\right]\longrightarrow0.
\]

\begin{corollary}\label{cor:capacity}
For every integer $K\ge2$ with $\delta_K>0$, let $n_0(K)$ be as in
Theorem~\ref{thm:main} and, for each integer $n\ge n_0(K)$, choose a
channel $T_{K,n}$ satisfying that theorem. Then
\begin{align}
 C(T_{K,n})-\chi(T_{K,n})
 &\ge\tfrac12n\delta_K-2\log\frac{n+1}{n-1}
 \qquad(n\ge n_0(K)),
 \label{eq:capacity-gap}\\
 \liminf_{n\to\infty}
 \frac{\chi(T_{K,n}^{\otimes2})}{2\chi(T_{K,n})}
 &\ge1+\frac{\delta_K}{2a_K}\ge\frac{\ln K}{18}.
 \label{eq:ratio}
\end{align}
The limit inferior in \eqref{eq:ratio} is taken with $K$ fixed.\footnote{An
unbounded ratio alone already follows from the one-block construction:
take $n=1$ and $\eta=K^{-2}$ in Proposition~\ref{prop:qualitative},
using Collins's estimate and Shor's extension \cite{Collins,Shor}.
The ratio is at least $(1+o(1))\ln K/18$, while the resulting additive-gap lower
bound is of order $(\log K)/K$ and tends to zero.}
In particular, for every $A>0$ and $R>0$ there is a finite-dimensional
channel $T$ with $\chi(T)>0$ such that
\[
 C(T)-\chi(T)\ge A,
 \qquad \frac{\chi(T^{\otimes2})}{2\chi(T)}\ge R.
\]
\end{corollary}
\begin{proof}
Equation \eqref{eq:capacity} gives
$C(T_{K,n})\ge\chi(T_{K,n}^{\otimes2})/2$, and \eqref{eq:main-gap} proves
\eqref{eq:capacity-gap}.
To justify division by $\chi(T_{K,n})$ in \eqref{eq:ratio}, note that
a channel with zero Holevo quantity has constant output, so its
tensor square also has zero Holevo quantity.%
\footnote{If $\chi(\Phi)=0$, distinct outputs would give positive
Holevo information for the equally weighted two-input ensemble,
by strict concavity of entropy. Thus $\Phi(\rho)=\sigma$ for a fixed
state $\sigma$ and every input state $\rho$. By linearity,
$\Phi(X)=(\Tr X)\sigma$, so
$(\Phi\otimes\Phi)(\rho)=\sigma\otimes\sigma$ for every joint input,
including entangled ones.}
But \eqref{eq:main-holevo} gives
$\chi(T_{K,n}^{\otimes2})\ge n(\log K)/K>0$, so
$\chi(T_{K,n})>0$. The first two bounds in Theorem~\ref{thm:main} give
\[
 \frac{\chi(T_{K,n}^{\otimes2})}{2\chi(T_{K,n})}
 \ge\frac{(\log K)/K}{2a_K+(4/n)\log\frac{n+1}{n-1}}.
\]
For each fixed $K$, take a limit inferior as $n\to\infty$ and use
\[
 1+\frac{\delta_K}{2a_K}=\frac{\log K}{2Ka_K}
 \ge\frac{\ln K}{18},
 \qquad a_K\le\frac9{K\ln2}.
\]
Given $A,R>0$, choose $K$ with $\delta_K>0$ and $\ln K/18>R$.
For sufficiently large $n$, \eqref{eq:ratio} and
\eqref{eq:capacity-gap} then give both required inequalities for
$T=T_{K,n}$.
\end{proof}

Allowing $K$ to grow gives vanishing single-use Holevo information
together with diverging two-use information. This already follows
from Proposition~\ref{prop:qualitative}, a qualitative version of
Theorem~\ref{thm:main} without an explicit input-dimension bound.

\begin{corollary}[Vanishing Holevo information and diverging capacity]
\label{cor:separation}
For every $\varepsilon>0$ and $R>0$ there is a finite-dimensional quantum
channel $T$ such that
\begin{equation}
 \chi(T)\le\varepsilon,
 \qquad \frac12\chi(T^{\otimes2})\ge R.
 \label{eq:small-large}
\end{equation}
In particular, there is a sequence of finite-dimensional channels
$T^{(K)}$ for which
\begin{equation}
 \chi(T^{(K)})\longrightarrow0,
 \qquad \frac12\chi\bigl((T^{(K)})^{\otimes2}\bigr)\longrightarrow\infty,
 \qquad C(T^{(K)})\longrightarrow\infty.
 \label{eq:vanishing-diverging}
\end{equation}
\end{corollary}

\begin{proof}
Let $n_K=\lceil K/\sqrt{\ln K}\rceil$ and $\eta_K=K^{-1}$.
Apply Proposition~\ref{prop:qualitative} separately for each integer
$K\ge2$, with $n=n_K$ and $\eta=\eta_K$, to obtain a channel
$T^{(K)}$ satisfying
\begin{align*}
 \chi(T^{(K)})
 &\le\frac{9n_K}{K\ln2}+K^{-1}\longrightarrow0,\\
 \frac12\chi\bigl((T^{(K)})^{\otimes2}\bigr)
 &\ge\frac{n_K\log K}{2K}
 \ge\frac{\sqrt{\ln K}}{2\ln2}\longrightarrow\infty.
\end{align*}
Here we used $a_K\le9/(K\ln2)$ and $n_K/K\to0$.
Equation~\eqref{eq:capacity} gives $C(T^{(K)})\to\infty$,
and choosing $K$ sufficiently large proves \eqref{eq:small-large}.
\end{proof}

\begin{remark}[Input-qubit cost of the separation]
\label{rem:separation-cost}
To obtain an explicit input cost, keep
$n_K=\lceil K/\sqrt{\ln K}\rceil$ but choose $T^{(K)}$ by the
construction \eqref{eq:channel-T} in the proof of
Theorem~\ref{thm:main}, with $n=n_K$.
Since $n_0(K)=O((\ln K)^2)$, we have $n_K/n_0(K)\to\infty$, so this is
possible for sufficiently large $K$.
For this choice, \eqref{eq:constructed-holevo-lower} and
\eqref{eq:main-holevo} give
\[
 \frac{2n_K}{K}\le\chi(T^{(K)})\le\frac{9n_K}{K\ln2}
       +2\log\frac{n_K+1}{n_K-1}\longrightarrow0,
\]
and the same two-use lower bound proves
\eqref{eq:vanishing-diverging}.
The correction term is $O(1/n_K)=o(n_K/K)$, and
$n_K/K\sim1/\sqrt{\ln K}$. The exact input dimension gives
$q_{\rm in}=(280/\ln2+o(1))K^2$, so
\[
 \chi(T^{(K)})=\Theta\!\left(\frac1{\sqrt{\log q_{\rm in}}}\right),
 \qquad
 \frac12\chi\bigl((T^{(K)})^{\otimes2}\bigr)
 =\Omega\!\left(\sqrt{\log q_{\rm in}}\right).
\]
\end{remark}

\paragraph{Related work.}
The additivity problem has been studied through channel reductions and
random-subspace methods \cite{AHW,MSW,HLW}, R\'enyi entropy counterexamples
\cite{WH,ASW10,CHLMW,GHP,SS,DL}, and further refinements of Hastings's
minimum-output-entropy counterexample \cite{BH,NS}.
For random-channel spectra and tensor powers, see
\cite{CN,CNII,BCN12,CFN,FN14,Montanaro,FN15} and the survey \cite{CNSurvey}.

Collins--Youn \cite{CY} proved regularized minimum-output-entropy
nonadditivity in an infinite-dimensional commuting-operator setting.
Lovitz--Wu \cite{LW} give random and deterministic permutation
realizations of the free-compression model, including a numerical
existence bound at output dimension $195$. They also discuss finite
approximation of the Collins--Youn construction in \cite[Section~1.3]{LW}.
Zhen--Zhu--Chen--Wang \cite{ZZCW} give a deterministic realization of
the Haagerup entropy bound for one factor using permutation matrices
restricted to the subspace orthogonal to the all-ones vector. They
also apply this construction to Holevo superadditivity.

The free-group estimate used here is Collins--Youn's
Theorem~3.3~\cite{CY} in our normalization, and its finite-dimensional
realization uses tensor-product strong convergence
\cite[Theorem~9.2]{BC}. Proposition~\ref{prop:qualitative} combines
these bounds with the Bell estimate and standard channel conversions
\cite{FW,Shor} to obtain unbounded finite-dimensional Holevo gaps,
controlling all inputs to each growing block. The quantitative
contribution is Proposition~\ref{prop:finite}: a simultaneous norm
bound for all required observables at $N=\lceil e^{b_Kn}\rceil$,
yielding the quadratic input-qubit cost in Theorem~\ref{thm:main}.

\paragraph{Concurrent work.}
Two works appeared shortly after ours. Shou and Gorshkov
construct a violation of $\smin$ additivity and extend it to large
gaps in v2 \cite{SG26}. Cheng and Wu give a deterministic version of
Hastings's random mixed-unitary setup using qudit Clifford unitaries,
a binary measurement, and classical feedforward \cite{CW26}. They also obtain the separation in
Corollary~\ref{cor:separation}: the product-state rate tends to zero
while the rate using entanglement within pairs grows without bound.
Haagerup's inequality is central to our proof and to Shou and
Gorshkov's. Cheng and Wu prove a finite-word analogue of the needed
norm bound and connect it to their explicit unitaries through moment
estimates. Both concurrent constructions are deterministic, while
ours uses random matrices. In both, the input-qubit cost is
exponential in the gap, whereas ours is quadratic.

\paragraph{Limitations and open questions.}
The scaling considered in \cite[Section~2, Eq.~(11)]{HP} requires both input and
output qubit counts to grow linearly with the
gap. Improving the quadratic input cost remains open.
The guaranteed gap per output qubit is extremely small, while the
sufficient input dimensions and the threshold $n_0(K)$ are very large.
The constants in these dimension bounds are not optimized; different
linearization methods and sharper random-matrix estimates could improve
them further.
The scalar Haagerup constant $3$ is asymptotically sharp
\cite[Section~7]{ZZCW}; improving the leading entropy-deficit
coefficient by this route requires additional structure, uniformly
over tensor powers.
The construction is probabilistic. Efficiently finding suitable unitaries and
verifying the required norm bound remain open problems.

\paragraph{On AI Usage.}
GPT-6 Astra (OpenAI) assisted with proof development and checking,
literature searches and reference verification, and the drafting and
revision of this manuscript. In particular, the model proposed the
method developed in Section~\ref{sec:approximation}. The author takes
responsibility for the mathematical claims, references, and final
presentation. The results and proofs are autoformalized by GPT-6 Astra in Lean, and checked against Comparator. See the \href{https://github.com/JWang226/Holevo-Additivity-Gap}{Lean Certificate}.

\section{Qualitative realization of unbounded Holevo gaps}
\label{sec:entropy}

The construction has three steps: realize the free-group norm bound
in finite dimensions, deduce the minimum-output-entropy gap, and
convert it into a two-use Holevo gap. The resulting
Proposition~\ref{prop:qualitative} leaves the input dimension
unspecified; Section~\ref{sec:approximation} supplies an explicit bound.

Fix integers $K\ge2$ and $n\ge1$. We write $B(H)$ for the bounded
operators on a Hilbert space $H$, and $\M_d=\M_d(\C)=B(\C^d)$.
Write $\norm A$ for operator norm,
and $\norm A_2=(\Tr A^*A)^{1/2}$ for unnormalized Hilbert--Schmidt norm.
The notation $[d]$ means $\{1,\ldots,d\}$.

\subsection{Finite-dimensional block channels and their entropy bound}
\label{sec:finite-entropy}

Following Collins \cite[Section~3.1]{Collins}, take unitaries
$W_1,\ldots,W_L\in U(D)$ and a branch register $\C^L$ with basis
$\{\ket a:a\in[L]\}$. The isometry
\[
 V:\C^D\longrightarrow\C^L\otimes\C^D,
 \qquad V\ket x=L^{-1/2}\sum_{a=1}^L\ket a\otimes W_a\ket x
\]
keeps the unitary label in superposition. Discarding that register
gives the mixed-unitary channel
\[
 R:\M_D\longrightarrow\M_D,
 \qquad R(\rho)=\frac1L\sum_{a=1}^L W_a\rho W_a^*.
\]
Discarding the other register gives the complementary channel
$F:\M_D\to\M_L$, with adjoint $F^*:\M_L\to\M_D$:
\begin{equation}
 F(\rho)=\frac1L\sum_{a,b=1}^L
       \Tr(\rho W_b^*W_a)\ket a\bra b,\qquad
 F^*(A)=\frac1L\sum_{a,b=1}^L\bra a A\ket b\,W_a^*W_b.
 \label{eq:gram}
\end{equation}
The adjoint transfers output observables to input observables through
$\Tr(F(\rho)A)=\Tr(\rho F^*(A))$. For a pure input, the two reduced
outputs have the same nonzero eigenvalues and hence the same entropy.

To form a block of $n$ factors, choose an integer $N\ge1$ and arbitrary unitaries
$U_{i,j}\in U(N)$, with $i\in[K]$ and $j\in[n]$, and set
\begin{equation}
 W_a=\bigotimes_{j=1}^n U_{a_j,j},\qquad a\in[K]^n.
 \label{eq:tensor-words}
\end{equation}
The complementary channel in factor $j$ is
\[
 F^{(j)}:\M_N\longrightarrow\M_K,\qquad
 F^{(j)}(\rho)=\frac1K\sum_{r,s=1}^K
 \Tr(\rho U_{s,j}^*U_{r,j})\ket r\bra s.
\]
The block channel is
$F_n=\bigotimes_{j=1}^n F^{(j)}:\M_{N^n}\to\M_{K^n}$, so
$D=N^n$ and $L=K^n$. Since $\smin(F_n)$ includes inputs
entangled across all $n$ factors, separate bounds for the $F^{(j)}$
do not suffice.

We obtain the block bound by comparing the matrices with unitary
shifts on $\F_K^n$, the direct product of one free
group for each channel factor.\footnote{The free group $\F_K$ consists of
reduced words in $K$ symbols $g_1,\ldots,g_K$ and their inverses.
Multiplication concatenates words and cancels adjacent inverse pairs;
the empty word is the identity. Reduced word length counts the symbols
remaining after cancellation. In $\F_K^n$, multiplication is performed
separately in each factor.}
On the orthonormal basis $\{\ket h:h\in\F_K^n\}$ of
$\ell^2(\F_K^n)$, the \emph{left regular representation}
acts by the unitary shifts
\[
 \lambda(g)\ket h=\ket{gh}.
\]
In this reference model, we replace $W_a$ by
\[
 w_a=(g_{a_1},\ldots,g_{a_n})\in\F_K^n,
 \qquad u_a=\lambda(w_a),\qquad a\in[K]^n.
\]
The subscript $\lambda$ indicates evaluation in this representation.
Under $\ell^2(\F_K^n)=\ell^2(\F_K)^{\otimes n}$, shifts from
different group factors act on separate tensor factors and commute.\footnote{The
canonical trace reads off the coefficient of the identity. With respect
to this trace, the generator shifts in each factor are freely independent
Haar unitaries \cite{NicaSpeicher}; the reference model is therefore a
tensor product of free Haar families.}

The product-group Haagerup inequality bounds the norm of a weighted sum
of these shifts by the Euclidean norm of its coefficients, with constants
determined by the word length in each group factor. As an operator norm
bound, it controls every vector, including entangled ones. The next
lemma applies Collins--Youn's estimate \cite[Theorem~3.3]{CY} to the
map obtained from $F_n^*$ by replacing $W_a$ with $u_a$, then
transfers the bound to finite matrices by tensor-product strong
convergence \cite[Theorem~9.2]{BC}.

\begin{lemma}[Free and finite-dimensional norm bounds]\label{lem:cy}
Let $K\ge2$ and $n\ge1$ be integers.
Define the linear map
$\Gamma_\lambda:\M_{K^n}\to B(\ell^2(\F_K^n))$ and the constant $c_n>0$ by
\begin{equation}
 \Gamma_\lambda(A)=\frac1{K^n}\sum_{a,b\in[K]^n}\bra a A\ket b\,u_a^*u_b,
 \qquad c_n^2=\frac{(1+9/K)^n-1}{K^n}.
 \label{eq:gamma-free}
\end{equation}
For every $A\in\M_{K^n}$ with $\Tr A=0$,
\begin{equation}
 \norm{\Gamma_\lambda(A)}\le c_n\norm A_2.
 \label{eq:cy}
\end{equation}

For every $\kappa>1$ and all sufficiently large integers $N$, there
exist tuples $(U_{1,j},\ldots,U_{K,j})\in U(N)^K$, $j\in[n]$, whose
block channel $F_n$ defined by \eqref{eq:gram}--\eqref{eq:tensor-words}
satisfies
\begin{equation}
 \norm{F_n^*(A)}\le\kappa c_n\norm A_2
 \quad\text{for every traceless Hermitian }A\in\M_{K^n}.
 \label{eq:norm-certificate}
\end{equation}
\end{lemma}
\begin{proof}
We first prove \eqref{eq:cy}. Collins--Youn's theorem
\cite[Theorem~3.3]{CY}, with their $N=K$ and $k=n$, gives
for every $A\in\M_{K^n}(\C)$ with $\Tr A=0$,
\[
 \left\|\sum_{a,b\in[K]^n}\bra a A\ket b\,
 \widetilde u_a^*\widetilde u_b\right\|
 \le K^{n/2}\sqrt{(1+9/K)^n-1}\,\norm A_2.
\]
Here
\[
 w_a=(g_{a_1},\ldots,g_{a_n})\in\F_\infty^n,
 \qquad \widetilde u_a=\widetilde\lambda(w_a),
\]
where $\widetilde\lambda$ denotes the left regular representation
of $\F_\infty^n$, and $g_1,g_2,\ldots$ are the canonical generators
of $\F_\infty$. Let
\[
 H:=\langle g_1,\ldots,g_K\rangle^n\le\F_\infty^n.
\]
Then $H\cong\F_K^n$, and every $w_a$, $a\in[K]^n$, belongs to $H$.
We identify $H$ with $\F_K^n$ and write $\lambda_H=\lambda$ for its
left regular representation.

To compare the regular representations, choose a set $\mathcal R$
of representatives for the right cosets $Ht$ of $H$ in
$\F_\infty^n$. Then
\[
 \ell^2(\F_\infty^n)
 =\bigoplus_{t\in\mathcal R}\ell^2(Ht).
\]
Each summand reduces the restriction of $\widetilde\lambda$ to $H$.
Indeed, the unitary
\[
 J_t:\ell^2(H)\longrightarrow\ell^2(Ht),
 \qquad J_t\ket h=\ket{ht},
\]
satisfies
\[
 \widetilde\lambda(g)J_t\ket h
 =\ket{ght}=J_t\lambda_H(g)\ket h,
 \qquad g,h\in H.
\]
Hence
\[
 \widetilde\lambda|_H
 \cong\bigoplus_{t\in\mathcal R}\lambda_H.
\]

Since $w_a^{-1}w_b\in H$ for all $a,b\in[K]^n$, it follows that
\[
 \sum_{a,b}\bra a A\ket b\,\widetilde u_a^*\widetilde u_b
 \cong\bigoplus_{t\in\mathcal R}
 \left(\sum_{a,b}\bra a A\ket b\,u_a^*u_b\right),
\]
where $u_a=\lambda_H(w_a)$. Therefore the two operators have the
same norm. Dividing the Collins--Youn estimate by $K^n$ yields
\[
 \left\|\frac1{K^n}\sum_{a,b}\bra a A\ket b\,u_a^*u_b\right\|
 \le K^{-n/2}\sqrt{(1+9/K)^n-1}\,\norm A_2
 =c_n\norm A_2,
\]
which is \eqref{eq:cy}.

For the finite-dimensional assertion, fix $\kappa>1$ and let
$\xi=(\kappa-1)/(\kappa+1)\in(0,1)$.
Fix a finite $\xi$-net $\mathcal T$ of the traceless Hermitian
$K^n\times K^n$ matrices of Hilbert--Schmidt norm one: every matrix
on this sphere is within distance $\xi$ of a member of $\mathcal T$.
Such a net exists by compactness and is independent of $N$.

Sample all $Kn$ unitaries independently from Haar measure, with each
$K$-tuple acting in its own tensor factor, and write $\Gamma_N=F_n^*$.
For each $A\in\mathcal T$, tensor-product strong convergence
\cite[Theorem~9.2]{BC} gives
\[
 \norm{\Gamma_N(A)}
 \xrightarrow[N\to\infty]{\mathrm{probability}}
 \norm{\Gamma_\lambda(A)}.
\]
Since the net is fixed and finite, the union bound and
\eqref{eq:cy} imply that, with probability tending to one,
$\norm{\Gamma_N(A)}\le(1+\xi)c_n$ simultaneously for all
$A\in\mathcal T$. For each sufficiently large $N$, choose such a
realization.

To extend the estimate to the whole sphere, set
\[
 Z=\sup_{A=A^*,\,\Tr A=0,\,\norm A_2=1}\norm{\Gamma_N(A)}.
\]
This supremum is finite. Approximating $A$ by a net point and using
linearity gives
\[
 Z\le(1+\xi)c_n+\xi Z,
 \qquad Z\le\frac{1+\xi}{1-\xi}c_n.
\]
Since $(1+\xi)/(1-\xi)=\kappa$, this proves
\eqref{eq:norm-certificate}.
\end{proof}

The next lemma converts the finite-dimensional adjoint norm bound
into an entropy bound by controlling the output's Hilbert--Schmidt
distance from $I_L/L$, and hence its purity.

\begin{lemma}[From an adjoint norm bound to output entropy]\label{lem:purity}
Let $F:\M_D\to\M_L$ be a channel. If $t\ge0$ and
\[
 \norm{F^*(A)}\le t\norm A_2
 \quad\text{for every traceless Hermitian }A\in\M_L,
\]
then every input state $\rho$ satisfies
\begin{equation}
 \Tr F(\rho)^2\le L^{-1}+t^2,
 \qquad \smin(F)\ge\log L-\log(1+Lt^2).
 \label{eq:purity}
\end{equation}
\end{lemma}
\begin{proof}
Let $Y=F(\rho)$ and $X=Y-I_L/L$. The matrix $X$ is traceless and
Hermitian. The defining identity for $F^*$ and the fact that $\rho$ is a state give
\[
 \norm X_2^2=\Tr(YX)=\Tr(\rho F^*(X))
 \le\norm{F^*(X)}\le t\norm X_2.
\]
Hence $\norm X_2\le t$, including the case $X=0$. Since
$\Tr Y^2=L^{-1}+\norm X_2^2$, this proves the purity bound. The entropy
bound follows from $S(Y)\ge-\log\Tr Y^2$, the comparison of von Neumann
entropy with R\'enyi entropy of order two.
\end{proof}

For any family satisfying \eqref{eq:norm-certificate} with
$\kappa\ge1$, we have
\[
 1+K^n\kappa^2c_n^2
 =\kappa^2(1+9/K)^n+1-\kappa^2
 \le\kappa^2(1+9/K)^n.
\]
With $a_K=\log(1+9/K)$, Lemma~\ref{lem:purity}, applied
with $L=K^n$ and $t=\kappa c_n$, gives
\begin{equation}
 \smin(F_n)\ge n(\log K-a_K)-2\log\kappa.
 \label{eq:single-entropy}
\end{equation}
The error $2\log\kappa$ is uniform over all block inputs and tends to
zero as $\kappa\to1$. Proposition~\ref{prop:finite} will give
$\kappa=(n+1)/(n-1)$ at $N=\lceil e^{b_Kn}\rceil$ for
$n\ge n_0(K)$.

\subsection{The minimum-output-entropy gap}

For the joint entropy upper bound, pair a channel $\Phi:\M_d\to\M_b$
with its conjugate $\overline\Phi(X)=\overline{\Phi(\overline X)}$
in fixed bases. The outputs $\Phi(\rho)$ and
$\overline\Phi(\overline\rho)=\overline{\Phi(\rho)}$ have the same
spectrum, so $\smin(\overline\Phi)=\smin(\Phi)$. The maximally entangled, or Bell,
state on $\C^d\otimes\C^d$ is
\[
 \ket{\omega_d}=d^{-1/2}\sum_{j=1}^d\ket{jj},
 \qquad\Omega_d=\ket{\omega_d}\bra{\omega_d}.
\]

The Bell-state argument for conjugate channels appears in Winter
\cite[Lemma~1]{Winter} and Hayden--Winter \cite[Lemma~II.1]{HW}; see also
Collins--Nechita \cite[Section~6.1]{CN}.
The von Neumann entropy bound used below is due to Hastings
\cite[Eq.~(10) and Supplemental Lemma~1]{Hastings}; see also
Collins \cite[Proposition~3.2]{Collins}.
We give the equal-weight argument (cf.\ \cite[Theorem~4.1]{CY})
and transfer it to complementary outputs and tensor blocks.

\begin{lemma}[Bell-output entropy bound]\label{lem:bell}
Let $F:\M_N\to\M_K$ be the complementary channel associated with
arbitrary unitaries $U_1,\ldots,U_K\in U(N)$ as in \eqref{eq:gram}.
Then
\begin{equation}
 S((F\otimes\overline F)(\Omega_N))
 \le2\log K-\frac{\log K}{K}.
 \label{eq:local-bell}
\end{equation}
Consequently every tensor-product family \eqref{eq:tensor-words} satisfies
\begin{equation}
 S((F_n\otimes\overline F_n)(\Omega_D))
 \le n\left(2\log K-\frac{\log K}{K}\right).
 \label{eq:bell}
\end{equation}
\end{lemma}
\begin{proof}
Let $R(\rho)=K^{-1}\sum_iU_i\rho U_i^*$ be complementary to $F$.
Since $(U_i\otimes\overline U_i)\ket{\omega_N}=\ket{\omega_N}$,
the $K$ diagonal pairs give
\[
 (R\otimes\overline R)(\Omega_N)
 =\frac1K\Omega_N+\frac1{K^2}\sum_{i\ne j}
 (U_i\otimes\overline U_j)\Omega_N(U_i\otimes\overline U_j)^*.
\]
This is a mixture of pure states with weights $1/K$ and
$K(K-1)$ copies of $1/K^2$, whose Shannon entropy is
$2\log K-(\log K)/K$. The output entropy is at most this value.
Under the tensor Stinespring isometry $V\otimes\overline V$,
the pure input $\Omega_N$ gives a pure state whose two reduced
outputs are $(R\otimes\overline R)(\Omega_N)$ and
$(F\otimes\overline F)(\Omega_N)$. Their entropies coincide,
proving \eqref{eq:local-bell}.
For the block channels, regroup the tensor factors so that
$\Omega_D=\bigotimes_{j=1}^n\Omega_N$. The output is then a tensor
product of the local Bell outputs, and entropy additivity gives
\eqref{eq:bell}.
\end{proof}

Combining the Bell estimate with \eqref{eq:single-entropy} gives
\begin{equation}
 2\smin(F_n)-\smin(F_n\otimes\overline F_n)
 \ge n\delta_K-4\log\kappa.
 \label{eq:moe}
\end{equation}
For fixed $K$ with $\delta_K>0$ and fixed $\kappa>1$, the finite
blocks constructed in Section~\ref{sec:finite-entropy} therefore have
minimum-output-entropy gaps tending to infinity as $n$ increases.
We now convert these gaps into two-use Holevo gaps of single channels.

\subsection{Converting the entropy gap into a Holevo gap}
\label{sec:conversion}

We apply the standard Weyl-covariant conversion, also used for one
factor in \cite[Section~6]{ZZCW}, to the whole block.
Combine $F_n$ and $\overline F_n$ using a switch register
$\C^2$ with basis $\ket0,\ket1$. Define
$\Psi_n:\M_2\otimes\M_D\to\M_L$ by
\[
 \Psi_n(\rho)=F_n\bigl((\bra0\otimes I_D)\rho(\ket0\otimes I_D)\bigr)
 +\overline F_n\bigl((\bra1\otimes I_D)\rho(\ket1\otimes I_D)\bigr).
\]
The measured switch selects $F_n$ or $\overline F_n$, then is
discarded. This is a variant of the direct-sum reduction of Fukuda--Wolf
\cite[Propositions~1--2]{FW} with the output branch label omitted;
the conjugate-channel switch is also described in
\cite[arXiv version, p.~2]{Hastings}. Complete positivity of the
summands and preservation of the total trace show that $\Psi_n$
is a channel.

The two branches have the same minimum output entropy. Since outputs
of $\Psi_n$ are mixtures of branch outputs, concavity and a fixed
switch value give
\[
 \smin(\Psi_n)=\smin(F_n).
\]
Fixing the two switches to $0$ and $1$ reproduces
$F_n\otimes\overline F_n$ on every joint input, including entangled
ones. Hence
\[
 \smin(\Psi_n^{\otimes2})
 \le\smin(F_n\otimes\overline F_n).
\]
The switch adds one input qubit and preserves the output dimension.
Discarding it is essential for Corollary~\ref{cor:separation}: an
output flag would transmit one bit without error and prevent
$\chi(T)\to0$.

Next, Shor's \emph{Weyl-covariant extension}
\cite[Section~9]{Shor} turns the general bound
$\chi(\Phi)\le\log L-\smin(\Phi)$ for $\Phi:\M_d\to\M_L$ into
an equality. It allows the sender to choose output conjugations
whose uniform average is maximally mixed.

In the basis $\{\ket j:j\in\mathbb Z_L\}$ of $\C^L$, define the shift
and phase unitaries by
\[
 X_L\ket j=\ket{j+1\pmod L},
 \qquad Z_L\ket j=e^{2\pi i j/L}\ket j.
\]
The $L^2$ Weyl unitaries $V_z=X_L^pZ_L^q$, indexed by
$z=(p,q)\in\mathbb Z_L^2$, satisfy the averaging identity
\begin{equation}
 L^{-2}\sum_zV_zYV_z^*=(\Tr Y)I_L/L.
 \label{eq:twirl}
\end{equation}
Add an input register $\C^{L^2}$ with basis
$\{\ket z:z\in\mathbb Z_L^2\}$ to select these unitaries. The extension
$\widetilde\Phi:\M_{L^2}\otimes\M_d\to\M_L$ is
\begin{equation}
 \widetilde\Phi(\rho)=\sum_zV_z
 \Phi\bigl((\bra z\otimes I_d)\rho(\ket z\otimes I_d)\bigr)V_z^*.
 \label{eq:extension}
\end{equation}
The expression inside $\Phi$ is the unnormalized input conditional
on the measured label $z$. The label is discarded after the
corresponding conjugation. The extension thus adds $2\log L$ input
qubits and preserves the output space $\C^L$.

Shor's construction \cite[Section~9]{Shor}, applied with independent
Weyl labels in each of the $m$ tensor factors, gives
\[
 \smin(\widetilde\Phi^{\otimes m})=\smin(\Phi^{\otimes m})
\]
and
\begin{equation}
 \chi(\widetilde\Phi^{\otimes m})
 =m\log L-\smin(\Phi^{\otimes m}),\qquad m\ge1.
 \label{eq:chi-extension}
\end{equation}
Conditioning on the labels gives the entropy identity; uniformly
varying them while fixing a minimizing joint input attains the Holevo
equality. The same argument for distinct channels gives the gap
identity in the introduction.

Apply the extension to $\Psi_n$, with input registers ordered as
Weyl label, switch, and quantum input. Writing
$\ket{zs}=\ket z\otimes\ket s$, the final channel is
\begin{equation}
 \begin{aligned}
 T&=\widetilde\Psi_n:\M_{L^2}\otimes\M_2\otimes\M_D\longrightarrow\M_L,\\
 T(\rho)&=\sum_zV_z\Bigl[
 F_n\bigl((\bra{z0}\otimes I_D)\rho(\ket{z0}\otimes I_D)\bigr)
 +\overline F_n\bigl((\bra{z1}\otimes I_D)\rho
                         (\ket{z1}\otimes I_D)\bigr)\Bigr]V_z^*.
 \end{aligned}
 \label{eq:channel-T}
\end{equation}
Since $D=N^n$ and $L=K^n$, its input dimension is $2N^nK^{2n}$ and
its output dimension is $K^n$. Equation~\eqref{eq:chi-extension} and
the switch identities give
\[
 \chi(T)=\log L-\smin(F_n),
 \qquad \chi(T^{\otimes2})\ge2\log L-
                                  \smin(F_n\otimes\overline F_n).
\]
Substituting \eqref{eq:single-entropy} and \eqref{eq:bell} yields
\begin{equation}
 \chi(T)\le na_K+2\log\kappa,
 \qquad \chi(T^{\otimes2})\ge n\frac{\log K}{K},
 \label{eq:certificate-holevo}
\end{equation}
and hence
\begin{equation}
 \Delta_\chi(T)\ge n\delta_K-4\log\kappa.
 \label{eq:certificate-gap}
\end{equation}
Only the upper bound on $\chi(T)$ and the resulting gap bound
require \eqref{eq:norm-certificate}; the two-use lower bound holds
for every choice of unitaries.

For the matching single-use estimate in Remark~\ref{rem:separation-cost},
we also need a lower bound on $\chi(T)$ that holds for every channel
constructed in \eqref{eq:channel-T}. In factor $j$, let
$R^{(j)}(\rho)=K^{-1}\sum_{i=1}^K U_{i,j}\rho U_{i,j}^*$
be the mixed-unitary channel complementary to $F^{(j)}$.
Since $U_{1,j}^*U_{2,j}$ is unitary, it has a unit eigenvector
$\ket{x_j}$ with
\[
 U_{1,j}^*U_{2,j}\ket{x_j}=e^{i\theta_j}\ket{x_j},
 \qquad
 U_{2,j}\ket{x_j}=e^{i\theta_j}U_{1,j}\ket{x_j}.
\]
Let $\rho_j=\ket{x_j}\bra{x_j}$ and
$\ket{y_{i,j}}=U_{i,j}\ket{x_j}$. The first two output vectors differ
only by a phase, so their rank-one projections coincide. Combining
their weights gives
\[
 R^{(j)}(\rho_j)
 =\frac2K\ket{y_{1,j}}\bra{y_{1,j}}
  +\frac1K\sum_{i=3}^K\ket{y_{i,j}}\bra{y_{i,j}}.
\]
The sum is empty when $K=2$. The entropy of a mixture of pure states
is at most the Shannon entropy of its weights, whether or not the
states are orthogonal. Here the weights are $2/K$ and $K-2$ copies
of $1/K$, so
\[
 \begin{aligned}
 S(R^{(j)}(\rho_j))
 &\le-\frac2K\log\frac2K-\frac{K-2}{K}\log\frac1K\\
 &=\frac2K(\log K-1)+\frac{K-2}{K}\log K
 =\log K-\frac2K.
 \end{aligned}
\]
The pure input $\rho_j$ gives a pure state under the local Stinespring
isometry. Its two reduced outputs are $R^{(j)}(\rho_j)$ and
$F^{(j)}(\rho_j)$, so their entropies agree.

Now use the product input $\rho=\bigotimes_{j=1}^n\rho_j$.
Since $F_n(\rho)=\bigotimes_{j=1}^nF^{(j)}(\rho_j)$, entropy
additivity for product states gives
\[
 \begin{aligned}
 \smin(F_n)
 &\le S(F_n(\rho))
 =\sum_{j=1}^n S(F^{(j)}(\rho_j))\\
 &=\sum_{j=1}^n S(R^{(j)}(\rho_j))
 \le n\left(\log K-\frac2K\right).
 \end{aligned}
\]
Substituting this upper bound into the exact identity
$\chi(T)=n\log K-\smin(F_n)$ yields
\begin{equation}
 \chi(T)\ge\frac{2n}{K}
 \qquad\text{for channels constructed in \eqref{eq:channel-T}.}
 \label{eq:constructed-holevo-lower}
\end{equation}

The construction and Lemma~\ref{lem:cy} give the following qualitative
form of Theorem~\ref{thm:main}.

\begin{proposition}[Qualitative finite-dimensional realization]\label{prop:qualitative}
For every pair of integers $K\ge2$, $n\ge1$ and every $\eta>0$, there
are an integer $N$ and a channel
$T:\M_{2N^nK^{2n}}(\C)\to\M_{K^n}(\C)$ such that
\begin{equation}
 \chi(T)\le na_K+\eta,
 \qquad \chi(T^{\otimes2})\ge n\frac{\log K}{K}.
 \label{eq:qualitative-bounds}
\end{equation}
Consequently,
\[
 \Delta_\chi(T)\ge n\delta_K-2\eta.
\]
\end{proposition}
\begin{proof}
Given $K,n,\eta$, let $\kappa=2^{\eta/2}>1$.
Lemma~\ref{lem:cy} supplies finite unitaries satisfying
\eqref{eq:norm-certificate} with this $\kappa$.
The channel \eqref{eq:channel-T} has the stated dimensions, and
\eqref{eq:certificate-holevo} gives \eqref{eq:qualitative-bounds}.
\end{proof}

Fixing $K$ with $\delta_K>0$ and $\eta>0$, apply the proposition
separately for each $n$. The resulting channels satisfy
$\Delta_\chi(T)\ge n\delta_K-2\eta\to\infty$; their matrix dimensions
may depend on $n$.

\section{Quantitative finite-dimensional approximation}
\label{sec:approximation}

We prove \eqref{eq:norm-certificate} with $\kappa=(n+1)/(n-1)$ and
$N=\lceil e^{b_Kn}\rceil$. The channel construction and entropy bounds
\eqref{eq:certificate-holevo}--\eqref{eq:certificate-gap} from
Section~\ref{sec:entropy} then give Theorem~\ref{thm:main}.

\begin{proposition}[Uniform finite-matrix estimate]\label{prop:finite}
Let $K\ge2$ be an integer, and let $n_0(K)$ and $b_K$ be as in
Theorem~\ref{thm:main}.
For every integer $n\ge n_0(K)$, let
\[
 N=\lceil e^{b_Kn}\rceil,
 \qquad c_n=\left(\frac{(1+9/K)^n-1}{K^n}\right)^{1/2},
 \qquad \kappa=\frac{n+1}{n-1}.
\]
There exist $n$ unitary tuples
$(U_{1,j},\ldots,U_{K,j})\in U(N)^K$, $1\le j\le n$, whose block
channel $F_n:\M_{N^n}(\C)\to\M_{K^n}(\C)$, defined by
\eqref{eq:gram}--\eqref{eq:tensor-words}, satisfies
\begin{equation}
 \norm{F_n^*(A)}\le\kappa c_n\norm A_2
 \quad\text{for every traceless Hermitian }A\in\M_{K^n}(\C).
 \label{eq:finite}
\end{equation}
\end{proposition}

\paragraph{Polynomials with matrix coefficients.}
The proof uses polynomials whose coefficients are fixed matrices.
For example, if $U,V\in U(N)$, then
\[
 \begin{pmatrix}1&0\\0&0\end{pmatrix}\otimes U
 +\begin{pmatrix}0&1\\0&0\end{pmatrix}\otimes U^*V
 +\begin{pmatrix}0&0\\0&1\end{pmatrix}\otimes V
 =\begin{pmatrix}U&U^*V\\0&V\end{pmatrix}.
\]
The coefficient matrices have size $2$, while the evaluated operator
has size $2N$. Its longest word is $U^*V$, of length two.

In general, for a unitary representation $\pi$ of a group on a Hilbert
space $\mathcal H_\pi$, a polynomial with matrix coefficients is a finite sum
\[
 R(\pi)=\sum_w a_w\otimes\pi(w),\qquad a_w\in\M_m(\C).
\]
It acts on $\C^m\otimes\mathcal H_\pi$. The \emph{coefficient size}
is $m$; the matrices $a_w$ stay fixed when the representation changes.
For $\pi_N$ on $(\C^N)^{\otimes n}$, the evaluated matrix has size
$mN^n$. The same coefficients can instead be evaluated in the regular
representation. A polynomial is \emph{linear} if its words are only
the identity and individual generators or their inverses.
Linearization will shorten the words by enlarging the auxiliary
coefficient space $\C^m$.

\begin{proof}[Proof of Proposition~\ref{prop:finite}]
We combine a finite net of observable tests into one polynomial,
express the generators using two base generators per factor, and
reduce the polynomial to degree one. The quantitative
Bordenave--Collins estimate gives a realization satisfying all tests
simultaneously; approximation from the net proves \eqref{eq:finite}.
The algebraic lemmas are proved in Appendix~\ref{app:operators};
Appendix~\ref{app:finite-threshold} gives the precise random-matrix
estimate and verifies the numerical bounds.

\subsection{Step 1: combine the observable tests}
\label{sec:finite-tests}

The estimate in \eqref{eq:finite} is homogeneous in $A$, so it is enough
to consider traceless Hermitian matrices with $\norm A_2=1$.
On this sphere, let
\[
 A_* =\frac{\ket1\bra2+\ket2\bra1}{\sqrt2}
       \otimes\left(\frac{I_K}{\sqrt K}\right)^{\otimes(n-1)}.
\]
Choose $\mathcal T_n$ maximal under inclusion among negation-closed
subsets of this sphere containing $\pm A_*$ and having pairwise
Hilbert--Schmidt distances $>1/n$. Such a finite set exists by compactness.
It is a $1/n$-net: any point $A$ farther than $1/n$ from the set could
be added together with $-A$, since the set is symmetric and
$\norm{A-(-A)}_2=2>1/n$, contradicting maximality.

The open balls of radius $1/(2n)$ centered at the points of
$\mathcal T_n$ are disjoint and lie in the ball of radius
$1+1/(2n)$ in the real vector space of dimension $K^{2n}-1$.
Comparing volumes gives
\begin{equation}
 J_n=|\mathcal T_n|\le(1+2n)^{K^{2n}-1}.
 \label{eq:net}
\end{equation}
The symmetry will give an exact norm identity in Step~3; the test
$A_*$ will control the first error factor through \eqref{eq:net-norm}.
Fix the tests before choosing any random matrices.

Use the group $\F_K^n$ and words $w_a$ from
Section~\ref{sec:entropy}. Evaluate the map \eqref{eq:gamma-free}
in a unitary representation $\pi$ and combine the tests by setting
\begin{equation}
 \Gamma_\pi(A)=\frac1{K^n}\sum_{a,b}\bra a A\ket b\,
                          \pi(w_a^{-1}w_b),
 \qquad P(\pi)=\bigoplus_{A\in\mathcal T_n}\Gamma_\pi(A).
 \label{eq:net-polynomial}
\end{equation}
This is a polynomial with diagonal matrix coefficients:
\[
 P(\pi)=\sum_{a,b}D_{ab}\otimes\pi(w_a^{-1}w_b),
 \qquad
 D_{ab}=\frac1{K^n}\operatorname{diag}_{A\in\mathcal T_n}
                         \bigl(\bra a A\ket b\bigr)\in\M_{J_n}(\C).
\]
Thus its coefficient size is $J_n$, and its norm is the largest norm
among the tests. Each block is self-adjoint, and the paired tests
$A,-A$ make the spectrum of $P(\pi)$ symmetric about zero.

Let $\lambda$ be the regular representation of $\F_K^n$.
Since every $A\in\mathcal T_n$ is traceless with $\norm A_2=1$,
the free-group bound \eqref{eq:cy} gives
\[
 \norm{P(\lambda)}
 =\max_{A\in\mathcal T_n}\norm{\Gamma_\lambda(A)}\le c_n.
\]
For the lower bound, use the prescribed test $A_*$ and let
$g=(g_1^{-1}g_2,1,\ldots,1)$ and $u=\lambda(g)$.
In \eqref{eq:net-polynomial}, the two off-diagonal entries of the
first factor of $A_*$ contribute $(u+u^*)/(\sqrt2K)$, while each
of its remaining $n-1$ factors $I_K/\sqrt K$ contributes
$K^{-1/2}I$. Thus
\[
 \Gamma_\lambda(A_*)
 =\frac{K^{-(n-1)/2}}{\sqrt2K}(u+u^*).
\]
The element $g$ has infinite order, since $(g_1^{-1}g_2)^r$ has
reduced-word length $2r$ for $r\ge1$. Hence the vectors $\ket{g^k}$
are orthonormal and $u\ket{g^k}=\ket{g^{k+1}}$.
For the unit vector $\xi_r=r^{-1/2}\sum_{k=0}^{r-1}\ket{g^k}$,
shifting the sum leaves $r-1$ overlapping terms, so
\[
 2\ge\norm{u+u^*}
 \ge\langle\xi_r,(u+u^*)\xi_r\rangle
 =\frac{2(r-1)}r\longrightarrow2.
\]
Thus $\norm{u+u^*}=2$. Since $A_*\in\mathcal T_n$, we obtain
\begin{equation}
 \sqrt2\,K^{-(n+1)/2}
 =\norm{\Gamma_\lambda(A_*)}\le\norm{P(\lambda)}\le c_n.
 \label{eq:net-norm}
\end{equation}

The lower bound in \eqref{eq:net-norm} controls error growth in Step~3,
allowing an explicit choice of approximation accuracy and dimension $N$.

For finite-dimensional unitary tuples, let $\pi_N$ denote their
representation on $(\C^N)^{\otimes n}$. Then
$\Gamma_{\pi_N}=\Gamma_N=F_n^*$, in the notation of
Section~\ref{sec:finite-entropy}. Our target is the single inequality
\begin{equation}
 \norm{P(\pi_N)}\le(1+1/n)\norm{P(\lambda)}.
 \label{eq:polynomial-target}
\end{equation}
Since $P(\pi_N)$ is the direct sum of the $F_n^*(A)$ over
$A\in\mathcal T_n$, the target bound and \eqref{eq:net-norm} give
\[
 \max_{A\in\mathcal T_n}\norm{F_n^*(A)}
 =\norm{P(\pi_N)}
 \le(1+1/n)\norm{P(\lambda)}
 \le(1+1/n)c_n.
\]
Step~4 extends this bound from the net to all traceless observables.

\subsection{Step 2: use two base generators}
\label{sec:two-generators}

The quantitative estimate deteriorates with the number of independent
Haar generators. We use only two in each tensor factor by realizing
the $K$ required generators as words in $\F_2=\langle x,y\rangle$.
Lemma~\ref{lem:embedding} gives words
$v_1,\ldots,v_K$, freely generating a copy of $\F_K$, with lengths at most $\ell=2\lfloor\log_2(K-1)\rfloor+1$.
The embedding preserves regular-representation norms, also for matrix
coefficients and tensor products. We henceforth identify $P$ and $w_a$
with their images in $\F_2^n$. If $\lambda_2$ denotes the regular
representation of $\F_2^n$, then $\norm{P(\lambda_2)}=\norm{P(\lambda)}$.

For each factor $j$, take independent Haar unitaries $X_j,Y_j\in U(N)$,
with all $2n$ matrices independent, and set\footnote{For $K=3$, one
can take $v_1=y$, $v_2=x^2$, and $v_3=xyx^{-1}$. Then
$U_{1,j}=Y_j$, $U_{2,j}=X_j^2$, and $U_{3,j}=X_jY_jX_j^*$.
These words are the free generators associated with the two-vertex
graph in Lemma~\ref{lem:embedding}, where $x$ interchanges the
vertices and $y$ has a loop at each vertex. Their lengths are at most
$\ell=3$.}
\[
 U_{i,j}=v_i(X_j,Y_j),\qquad 1\le i\le K,\quad 1\le j\le n.
\]
Here $i$ labels the unitary and $j$ the tensor factor. The $K$ derived
unitaries within one factor need not be independent. The base pairs
give a representation $\pi_N$ of $\F_2^n$;
its restriction to $\F_K^n$ is the representation used in
Step~1 and defines the channel in the proposition.

The words in $P$ have length at most $2\ell n$, where length is the sum
of the reduced-word lengths in the $n$ factors.

\subsection{Step 3: reduce to a linear polynomial}
\label{sec:degree-reduction}

Lemma~\ref{lem:haar} applies to self-adjoint linear polynomials.
We therefore shorten $P$ in stages to such a polynomial $H$, enlarging
the matrix coefficients at each stage. Exact norm identities will
transfer the estimate for $H$ back to $P$.

\paragraph{A minimal example.}
For any two unitaries $U,V$ on the same Hilbert space, consider
the quadratic polynomial
\[
 R_{\mathrm{ex}}=U^*V+V^*U.
\]
Define the linear polynomial
\[
 Q_{\mathrm{ex}}
 =\begin{pmatrix}U+V&0\\0&U-V\end{pmatrix}
 =I_2\otimes U+\begin{pmatrix}1&0\\0&-1\end{pmatrix}\otimes V.
\]
Its coefficients are fixed $2\times2$ matrices. Direct multiplication gives
\[
 Q_{\mathrm{ex}}^*Q_{\mathrm{ex}}
 =\begin{pmatrix}2I+R_{\mathrm{ex}}&0\\0&2I-R_{\mathrm{ex}}\end{pmatrix},
 \qquad
 \norm{Q_{\mathrm{ex}}}^2=2+\norm{R_{\mathrm{ex}}}.
\]
The two diagonal blocks capture both signs of the self-adjoint
operator $R_{\mathrm{ex}}$. Thus words of length two have been replaced
by single letters, the coefficient size has increased from one to two,
and the original norm is recovered by squaring the new norm and
subtracting two. This holds for every choice of $U,V$.

\paragraph{The factorization rule.}
Lemma~\ref{lem:linear} gives the same type of norm identity for a general
polynomial $R$. Choose a finite set $S$ containing the identity such
that every word in $R$ can be written as
\[
 w=g^{-1}h,\qquad g,h\in S.
\]
Thus the words of $R$ lie in $S^{-1}S$. If its coefficients belong to
$\M_m(\C)$ and $B=|S|$, the lemma constructs
\[
 Q(\pi)=\sum_{g\in S}b_g\otimes\pi(g),\qquad
 b_g\in\M_{2mB}(\C),
\]
and a constant $\theta\ge0$ such that
\[
 \norm{Q(\pi)}^2=\norm{R(\pi)}+\theta
 \quad\text{for every }\pi,\qquad
 \theta\le B\norm{R(\lambda_2)}.
\]
The new polynomial uses only the words in $S$. Our task is therefore
to choose $S$ so that its words are shorter than those of $R$.
For a word $w=ab$, the choice $g=a^{-1}$ and $h=b$ gives
$w=g^{-1}h$; this explains why the sets below include inverses.

The lemma uses the Hermitian dilation of $R$, so it must also handle
the words in the adjoint:
\[
 R(\pi)^*=\sum_w a_w^*\otimes\pi(w^{-1}),\qquad
 w=g^{-1}h\ \Longrightarrow\ w^{-1}=h^{-1}g.
\]
The same set $S$ works for these words, with $g$ and $h$ interchanged.
The identity in $S$ allows a word that is already short enough to
pass through a later step as $w=1^{-1}w$.

The norm identity also quantifies the error when we work backward.
Suppose $0\le e\le1$ and
\[
 \norm{Q(\pi_N)}\le(1+e)\norm{Q(\lambda_2)}.
\]
Then
\[
 \begin{aligned}
 \norm{R(\pi_N)}-\norm{R(\lambda_2)}
 &\le\bigl((1+e)^2-1\bigr)
          \bigl(\norm{R(\lambda_2)}+\theta\bigr)\\
 &\le3e(1+B)\norm{R(\lambda_2)}\\
 &\le6Be\,\norm{R(\lambda_2)},
 \end{aligned}
\]
where the last two inequalities use the bound on $\theta$,
$2e+e^2\le3e$, and $B\ge1$.
Thus the coefficient size is multiplied by $2B$ in the forward
construction, and the relative error is multiplied by at most
$6B$ in the backward estimate.
All reference norms encountered below are positive, by
\eqref{eq:net-norm} and the norm identities.

\paragraph{The stages of the construction.}
We use the special form of $P$ for the initial reduction and then
apply the factorization rule repeatedly. Let
\[
 q=\lceil\log_2n\rceil,\qquad s=\lceil\log_2\ell\rceil.
\]
We denote the polynomial after the initial reduction by $Q^{(0)}$,
after $j$ divisions of the tensor factors by $Q^{(j)}$, and after
$i$ further word-shortening steps by $Q^{(q+i)}$.
Let $d_r$ be the coefficient size of $Q^{(r)}$. For backward estimates,
define the nonnegative relative error at this stage by
\[
 e_r=\max\left\{0,\frac{\norm{Q^{(r)}(\pi_N)}}
                         {\norm{Q^{(r)}(\lambda_2)}}-1\right\}.
\]
The final Hermitian dilation gives $H$.
Figure~\ref{fig:linearization} records the words allowed at each stage.
All coefficients and constants are determined by the fixed net and
the regular representation, independently of the sampled matrices.

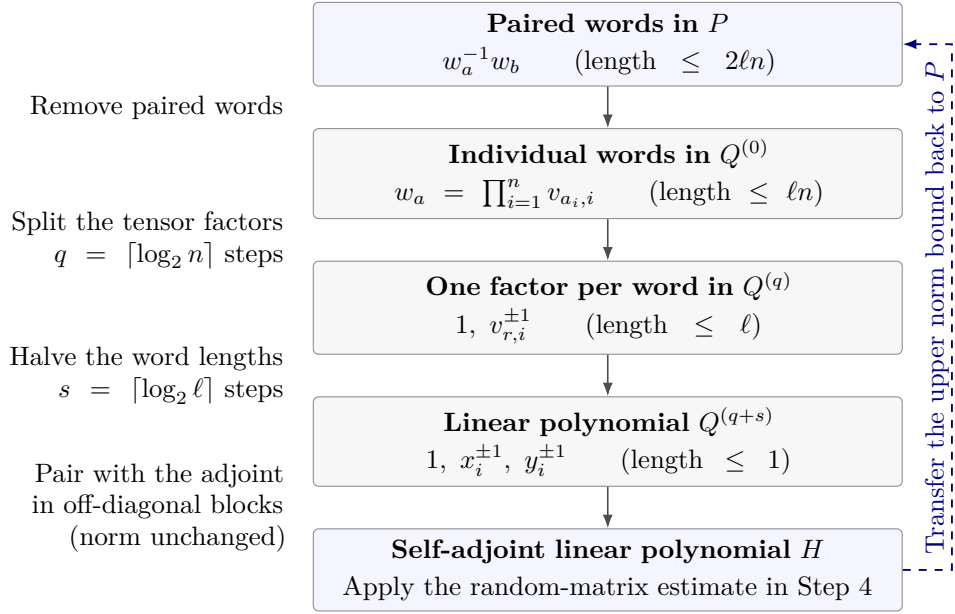
\begin{figure}[!htbp]
\centering
\begin{tikzpicture}[
 stage/.style={draw=black!45,fill=black!3,rounded corners=2pt,
   text width=7.5cm,minimum height=1.02cm,inner sep=4pt,
   align=center,font=\small},
 action/.style={text width=4.1cm,align=right,font=\small},
 forward/.style={-{Latex[length=1.8mm]},draw=black!70,line width=.6pt},
 backward/.style={-{Latex[length=1.8mm]},draw=blue!50!black,
   dashed,line width=.7pt}
]
\node[stage,fill=blue!4] (paired)
 {\textbf{Paired words in $P$}\\[2pt]
  $w_a^{-1}w_b\qquad (\text{length}\le2\ell n)$};
\node[stage,below=.55cm of paired] (single)
 {\textbf{Individual words in $Q^{(0)}$}\\[2pt]
  $w_a=\prod_{i=1}^n v_{a_i,i}\qquad (\text{length}\le\ell n)$};
\node[stage,below=.55cm of single] (factor)
 {\textbf{One factor per word in $Q^{(q)}$}\\[2pt]
  $1,\ v_{r,i}^{\pm1}\qquad (\text{length}\le\ell)$};
\node[stage,below=.55cm of factor] (linear)
 {\textbf{Linear polynomial $Q^{(q+s)}$}\\[2pt]
  $1,\ x_i^{\pm1},\ y_i^{\pm1}\qquad (\text{length}\le1)$};
\node[stage,fill=blue!4,below=.55cm of linear] (hermitian)
 {\textbf{Self-adjoint linear polynomial $H$}\\[2pt]
  Apply the random-matrix estimate in Step~4};

\draw[forward] (paired.south) -- (single.north);
\draw[forward] (single.south) -- (factor.north);
\draw[forward] (factor.south) -- (linear.north);
\draw[forward] (linear.south) -- (hermitian.north);

\node[action,anchor=east] at
 ($ (paired.south)!.5!(single.north)+(-4.15cm,0) $)
 {Remove paired words};
\node[action,anchor=east] at
 ($ (single.south)!.5!(factor.north)+(-4.15cm,0) $)
 {Split the tensor factors\\$q=\lceil\log_2n\rceil$ steps};
\node[action,anchor=east] at
 ($ (factor.south)!.5!(linear.north)+(-4.15cm,0) $)
 {Halve the word lengths\\$s=\lceil\log_2\ell\rceil$ steps};
\node[action,anchor=east] at
 ($ (linear.south)!.5!(hermitian.north)+(-4.15cm,0) $)
 {Pair with the adjoint\\in off-diagonal blocks\\(norm unchanged)};

\coordinate (returntop) at ($(paired.east)+(.65cm,0)$);
\coordinate (returnbottom) at (returntop |- hermitian.east);
\draw[backward] (hermitian.east) -- (returnbottom)
 -- node[midway,right=.10cm,rotate=90,anchor=south,font=\small,
         text=blue!50!black]
 {Transfer the upper norm bound back to $P$}
 (returntop) -- (paired.east);
\end{tikzpicture}
\caption{Allowed words at each stage of linearization. Here $r$
labels one of the $K$ words, and $x_i,y_i$ are the two base generators
in tensor factor $i$. The forward construction enlarges matrix
coefficients; the return arrow shows how Step~4 transfers the norm
bound from $H$ to $P$. The shortening stage is skipped when $\ell=1$.}
\label{fig:linearization}
\end{figure}

\paragraph{From paired words to individual words.}
For each test $A\in\mathcal T_n$, the block
\[
 \Gamma_\pi(A)=\frac1{K^n}\sum_{a,b}\bra a A\ket b\,
                             \pi(w_a)^*\pi(w_b)
\]
is a sum of paired words. Number the $K^n$ words by $1,\ldots,K^n$
and place their operators in the first column:
\[
 V(\pi)=\sum_a\ket a\bra1\otimes\pi(w_a)
 =\begin{pmatrix}
 \pi(w_1)&0&\cdots&0\\
 \vdots&\vdots&&\vdots\\
 \pi(w_{K^n})&0&\cdots&0
 \end{pmatrix}.
\]
The matrices acting on the row labels are the auxiliary coefficients;
the entries $\pi(w_a)$ act on the representation space.
Since $A$ is Hermitian and $\norm A\le\norm A_2=1$, define
\[
 Q_A(\pi)=K^{-n/2}\bigl((I_{K^n}+A)^{1/2}\otimes I\bigr)V(\pi),
 \qquad Q^{(0)}(\pi)=\bigoplus_{A\in\mathcal T_n}Q_A(\pi).
\]
Multiplication by the fixed square root changes the coefficients
but leaves only the individual words $w_a$. Moreover,
\[
 \begin{aligned}
 Q_A(\pi)^*Q_A(\pi)
 &=K^{-n}V(\pi)^*\bigl((I_{K^n}+A)\otimes I\bigr)V(\pi)\\
 &=\ket1\bra1\otimes\bigl(I+\Gamma_\pi(A)\bigr).
 \end{aligned}
\]
Here the identity term comes from
$K^{-n}\sum_a\pi(w_a)^*\pi(w_a)=I$.
As in the minimal example, both signs are needed to recover a norm.
The net contains $A$ and $-A$, and
$\max\{\norm{I+B},\norm{I-B}\}=1+\norm B$ for self-adjoint $B$.
Consequently,
\begin{equation}
 \norm{Q^{(0)}(\pi)}^2
 =\max_{A\in\mathcal T_n}\norm{I+\Gamma_\pi(A)}
 =1+\norm{P(\pi)}.
 \label{eq:gram-reduction}
\end{equation}
There are $J_n$ blocks, each with coefficient size $K^n$, so
$d_0=J_nK^n$. This special construction has additive constant one
and avoids the extra coefficient factor that a direct application of
Lemma~\ref{lem:linear} would introduce.

For the corresponding backward estimate, if $0\le e\le1$ and
$\norm{Q^{(0)}(\pi_N)}\le(1+e)\norm{Q^{(0)}(\lambda_2)}$, then
\[
 \begin{aligned}
 \norm{P(\pi_N)}-\norm{P(\lambda)}
 &\le\bigl((1+e)^2-1\bigr)(1+\norm{P(\lambda)})\\
 &\le3e(1+\norm{P(\lambda)}).
 \end{aligned}
\]
We used $\norm{P(\lambda_2)}=\norm{P(\lambda)}$.
Dividing by the positive reference norm gives the initial
relative-error factor $3(1+\norm{P(\lambda)}^{-1})$.

\paragraph{From all tensor factors to one factor per word.}
The words in $Q^{(0)}$ are
\[
 w_a=v_{a_1,1}\cdots v_{a_n,n},\qquad
 \pi_N(w_a)=U_{a_1,1}\otimes\cdots\otimes U_{a_n,n}.
\]
Here $a_i\in[K]$ selects a word in factor $i$.
We now reduce the number of factors occurring in each word by
splitting groups of factors into two.

For $n=4$, the groups change as
\[
 \{1,2,3,4\}
 \longrightarrow\{1,2\},\{3,4\}
 \longrightarrow\{1\},\{2\},\{3\},\{4\}.
\]
At the first division, let $S_1$ contain the identity, all words
$v_{r,1}v_{t,2}$ and $v_{r,3}v_{t,4}$ with $r,t\in[K]$, and their
inverses. Then
\[
 w_a=
 \underbrace{v_{a_1,1}v_{a_2,2}}_{g^{-1}}\,
 \underbrace{v_{a_3,3}v_{a_4,4}}_{h},
 \qquad
 g=(v_{a_1,1}v_{a_2,2})^{-1},\quad
 h=v_{a_3,3}v_{a_4,4}
\]
has $g,h\in S_1$. The factorization rule therefore replaces
$Q^{(0)}$ by $Q^{(1)}$ using words from at most two factors.
For the next division, use
\[
 S_2=\{1\}\cup\{v_{r,i}^{\pm1}:r\in[K],\ i\in[4]\}.
\]
For instance,
$v_{a_1,1}v_{a_2,2}=(v_{a_1,1}^{-1})^{-1}v_{a_2,2}$;
the inverse word is handled by interchanging the two elements of
$S_2$. The resulting $Q^{(2)}$ uses at most one factor per word.

For general $n$, bisect every group of size greater than one into
parts whose sizes differ by at most one, retaining singletons.
At level $j$, define
\[
 S_j=\{1\}\cup
 \left\{\prod_{i\in I}v_{a_i,i}^{\varepsilon}:
 I\text{ a group at level }j,\ (a_i)\in[K]^I,\
 \varepsilon\in\{1,-1\}\right\}.
\]
A word on a group that splits is a product $ab$ over its two new
groups; choose $g=a^{-1}$ and $h=b$. Both lie in $S_j$.
Words from distinct factors commute, so
\[
 \left(\prod_{i\in I}v_{a_i,i}\right)^{-1}
 =\prod_{i\in I}v_{a_i,i}^{-1}.
\]
This explains the common sign in $S_j$ and ensures that the same
splitting handles inverse words. For a retained singleton use
$w=1^{-1}w$, and for the identity use $1=1^{-1}1$.
Thus all words of $Q^{(j-1)}$ lie in $S_j^{-1}S_j$.

There are at most $2^j$ groups, with at most $\lceil n/2^j\rceil$
factors in each. Counting choices of words and both signs gives
\begin{equation}
 B_j=|S_j|
 \le1+\underbrace{2^j}_{\text{groups}}
       \underbrace{K^{\lceil n/2^j\rceil}}_{\text{choices per group}}
       \underbrace{2}_{\text{signs}}
 =1+2^{j+1}K^{\lceil n/2^j\rceil},
 \quad 1\le j\le q.
 \label{eq:coordinate-support}
\end{equation}
Applying the rule with $R=Q^{(j-1)}$, $Q=Q^{(j)}$, and $S=S_j$
gives
\[
 \begin{gathered}
 d_j=2B_jd_{j-1},\qquad
 \norm{Q^{(j)}(\pi)}^2
     =\norm{Q^{(j-1)}(\pi)}+\theta_j,\\
 0\le\theta_j\le B_j\norm{Q^{(j-1)}(\lambda_2)}.
 \end{gathered}
\]
The backward calculation above gives
\[
 e_{j-1}\le6B_je_j,\qquad 0\le e_j\le1.
\]
After $q$ divisions, $Q^{(q)}$ uses only
$1$ and $v_{r,i}^{\pm1}$, each acting in at most one factor.

\paragraph{From words within one factor to single letters.}
Each remaining word has length at most $\ell$ in the letters
$x_j,x_j^{-1},y_j,y_j^{-1}$ of one factor. We use the same
factorization rule, now splitting the sequence of letters.

At shortening step $i$, let $r=\lceil\ell/2^i\rceil$ and let
$\mathcal S_i$ contain the identity and all reduced words of length
at most $r$ in any one factor. A reduced word has no adjacent
inverse pair. The current words have length at most
\[
 \left\lceil\frac{\ell}{2^{i-1}}\right\rceil
 \le2\left\lceil\frac{\ell}{2^i}\right\rceil=2r.
\]
Cut any such word into $w=ab$ with $|a|,|b|\le r$. Then
\[
 w=(a^{-1})^{-1}b,\qquad a^{-1},b\in\mathcal S_i.
\]
Both pieces remain in the same factor, and taking an inverse preserves
reduced-word length. Empty pieces are the identity.
Thus this step requires no commutation within a factor.

For example, when $\ell=5$ the length bounds are
\[
 5\longrightarrow3\longrightarrow2\longrightarrow1.
\]
A first-step cut can be
\[
 w=\underbrace{x_jy_jx_j^{-1}}_{a}\,
   \underbrace{y_jx_j}_{b},\qquad
 a^{-1}=x_jy_j^{-1}x_j^{-1}.
\]
The words $a^{-1}$ and $b$ have lengths three and two, respectively,
so both belong to $\mathcal S_1$. Including all reduced words up to
the current bound ensures that the next step also applies to the
newly produced polynomial.

There are $4\cdot3^{t-1}$ reduced words of length $t\ge1$ in one
factor: four choices for the first letter and three thereafter.
Different factors share only the identity. Hence
\[
 C_i=|\mathcal S_i|
 =1+n\sum_{t=1}^r4\cdot3^{t-1}
 =1+2n(3^r-1)\le2n3^r,
\]
or
\begin{equation}
 C_i\le2n3^{\lceil\ell/2^i\rceil},\qquad 1\le i\le s.
 \label{eq:word-counts}
\end{equation}
The factorization with $S=\mathcal S_i$ replaces
$Q^{(q+i-1)}$ by $Q^{(q+i)}$, with
\[
 d_{q+i}=2C_i d_{q+i-1},\qquad
 e_{q+i-1}\le6C_ie_{q+i}
 \quad(0\le e_{q+i}\le1).
\]
After $s$ steps, the length bound is $\lceil\ell/2^s\rceil=1$.
If $\ell=1$, this part of the construction is omitted.

\paragraph{The self-adjoint linear polynomial and its total cost.}
The only words in $Q^{(q+s)}$ are now
$1,x_j^{\pm1},y_j^{\pm1}$ for $j\in[n]$.
Let
\[
 H(\pi)=
 \begin{pmatrix}
 0&Q^{(q+s)}(\pi)\\ Q^{(q+s)}(\pi)^*&0
 \end{pmatrix}.
\]
This Hermitian dilation is self-adjoint and linear, since taking the
adjoint only inverts the single letters. Squaring it gives
\[
 H(\pi)^2=\operatorname{diag}\bigl(
 Q^{(q+s)}(\pi)Q^{(q+s)}(\pi)^*,
 Q^{(q+s)}(\pi)^*Q^{(q+s)}(\pi)\bigr),
 \qquad \norm{H(\pi)}=\norm{Q^{(q+s)}(\pi)}.
\]
Thus the final operation doubles the coefficient size and introduces
no norm error. Writing $m_n=2d_{q+s}$ for the coefficient size of $H$,
we have
\begin{align}
 m_n&\le2J_nK^n\prod_{j=1}^q(2B_j)\prod_{i=1}^s(2C_i),
 \label{eq:coefficient-product}\\
 \mathfrak T_n&:=3(1+\norm{P(\lambda)}^{-1})
             \prod_{j=1}^q(6B_j)\prod_{i=1}^s(6C_i).
 \label{eq:error-product}
\end{align}
These products record the coefficient size of $H$ and the total
relative-error factor when returning from $H$ to $P$.

Let
\begin{equation}
 \gamma_K=\tfrac32\ln K+\tfrac12,\qquad
 \varepsilon_n=\frac{e^{-\gamma_Kn}}n.
 \label{eq:tolerance}
\end{equation}
Appendix~\ref{app:reduction-costs} sums the logarithms of these products
and proves, for every $n\ge n_0(K)$,
\begin{equation}
 \mathfrak T_n\le e^{\gamma_Kn},\qquad
 \ln(2m_n)\le2K^{2n}\ln(1+2n).
 \label{eq:error-size}
\end{equation}
Thus $\mathfrak T_n\varepsilon_n\le1/n$: a relative norm error
$\varepsilon_n$ for $H$ will imply \eqref{eq:polynomial-target}.
Step~4 obtains this estimate for $H$ and checks that each intermediate
error remains at most one, as required by the backward calculation.

\subsection{Step 4: apply the random-matrix estimate}
\label{sec:finish-finite}

For the self-adjoint linear polynomial $H$, Lemma~\ref{lem:haar}
compares its random value $H(\pi_N)$ with its regular-representation
value $H(\lambda_2)$:
\begin{equation}
 \E\norm{H(\pi_N)}\le\norm{H(\lambda_2)}
 (2m_n)^{(n+1)/p}N^{n(n+1)/(2p)}
 p^{3n(n+1)/(2p)}(1+N^{-1/2})^n,
 \label{eq:haar}
\end{equation}
provided $p$ is an even integer with
$2\le p\le2^{-2/5}N^{1/80}$. There is no restriction on the coefficient
size $m_n$, whose contribution to the logarithm of the bound is only
$(n+1)\ln(2m_n)/p$, which is why we can include all the tests in one
polynomial.

Let
\[
 N=\lceil e^{b_Kn}\rceil,\qquad
 p=2\left\lfloor\frac{N^{1/80}}4\right\rfloor.
\]
For $K\ge2$ and $n\ge2$, we have $N^{1/80}\ge K^{7n/2}\ge2^7$.
Our choice is the largest even integer at most $N^{1/80}/2$, so
\[
 2\le\frac{N^{1/80}}4
 \le\frac{N^{1/80}}2-2\le p\le\frac{N^{1/80}}2
 <2^{-2/5}N^{1/80}.
\]
Write the factor multiplying $\norm{H(\lambda_2)}$ in
\eqref{eq:haar} as $e^{E_n}$, where
\begin{equation}
 E_n=\frac{(n+1)\ln(2m_n)}p
 +\frac{n(n+1)\ln N}{2p}
 +\frac{3n(n+1)\ln p}{2p}
 +n\ln(1+N^{-1/2}).
 \label{eq:En}
\end{equation}
The choice $b_K/80=2\ln K+\gamma_K+1/2$ gives
$p^{-1}K^{2n}e^{\gamma_Kn}\le4e^{-n/2}$.
Thus the decay of $1/p$ offsets the coefficient growth and required
accuracy, with an exponential margin to absorb polynomial factors.
Appendix~\ref{app:haar-error} makes this comparison explicit and proves
\begin{equation}
 E_n\le\varepsilon_n/16<\ln(1+\varepsilon_n/2),
 \qquad n\ge n_0(K).
 \label{eq:En-bound}
\end{equation}
Consequently,
$\E\norm{H(\pi_N)}\le(1+\varepsilon_n/2)\norm{H(\lambda_2)}$.
The norm $\norm{H(\lambda_2)}$ is positive: $\norm{P(\lambda)}>0$ by
\eqref{eq:net-norm}, and every factorization preserves positivity of
the reference norm. Markov's inequality now gives
\begin{equation}
 \Prob\{\norm{H(\pi_N)}\le(1+\varepsilon_n)\norm{H(\lambda_2)}\}
 \ge\frac{\varepsilon_n}{2(1+\varepsilon_n)}>0.
 \label{eq:positive-probability}
\end{equation}

Fix a realization in this event. Removing the final Hermitian
dilation preserves the norm, giving relative error at most
$\varepsilon_n$ for $Q^{(q+s)}$.
Undo the shortening steps in the order $s,\ldots,1$, the tensor
splits in the order $q,\ldots,1$, and the initial reduction to $P$.
The error-transfer bounds from Step~3 give the following table.
The induction below checks their hypothesis $e\le1$ at each step.
\begin{center}
\begin{tabular}{@{}ll@{}}
\toprule
Reductions undone & Relative-error bound afterward \\
\midrule
Shortening steps $s,\ldots,i$ ($1\le i\le s$)
 & $\displaystyle\varepsilon_n\prod_{t=i}^s(6C_t)$ \\[8pt]
Tensor splits $q,\ldots,j$ ($1\le j\le q$)
 & $\displaystyle\varepsilon_n\prod_{t=j}^q(6B_t)
                      \prod_{i=1}^s(6C_i)$ \\[8pt]
Initial reduction, returning to $P$
 & $\displaystyle\mathfrak T_n\varepsilon_n$ \\
\bottomrule
\end{tabular}
\end{center}
The last row follows by multiplying the preceding bound at $j=1$
by $3(1+\norm{P(\lambda)}^{-1})$, as in \eqref{eq:error-product}.
If $s=0$, the shortening row is absent and its product equals one.

We justify the table by induction. Initially the error bound is
$\varepsilon_n$, corresponding to the empty product. Suppose the
current bound is $\varepsilon_n$ times the factors already used.
Every factor in \eqref{eq:error-product} is at least one, so their
partial product is at most $\mathfrak T_n$. The error entering the
next reversal therefore satisfies
\[
 e\le\mathfrak T_n\varepsilon_n\le1/n\le1/2<1.
\]
The next error-transfer bound applies and contributes the next
factor in the table, completing the induction.
At the end, using $\norm{P(\lambda_2)}=\norm{P(\lambda)}$, we obtain
\[
 \norm{P(\pi_N)}
 \le(1+\mathfrak T_n\varepsilon_n)\norm{P(\lambda)}
 \le(1+1/n)\norm{P(\lambda)}.
\]
This proves \eqref{eq:polynomial-target} and, by \eqref{eq:net-norm}, gives
$\norm{F_n^*(B)}\le(1+1/n)c_n$ for every $B\in\mathcal T_n$.

To extend the bound to all traceless observables, let
\[
 Z=\sup_{A=A^*,\,\Tr A=0,\,\norm A_2=1}\norm{F_n^*(A)}.
\]
This number is finite by compactness. Given any $A$ in this supremum,
choose $B\in\mathcal T_n$ with $\norm{A-B}_2\le1/n$.
Since $A-B$ is traceless and Hermitian, homogeneity gives
\[
 \norm{F_n^*(A)}
 \le\norm{F_n^*(B)}+\norm{F_n^*(A-B)}
 \le(1+1/n)c_n+Z/n.
\]
Taking the supremum and rearranging yields
\[
 Z\le\frac{n+1}{n-1}c_n=\kappa c_n.
\]
Rescaling $A$ proves \eqref{eq:finite} and completes the proposition.
\end{proof}

\begin{proof}[Proof of Theorem~\ref{thm:main}]
Choose the unitary tuples from Proposition~\ref{prop:finite} and apply
the input switch and Weyl extension from Section~\ref{sec:conversion}.
The resulting channel is $T_n=T$ in \eqref{eq:channel-T}, with input
dimension $2N^nK^{2n}$ and output dimension $K^n$.
The bound \eqref{eq:finite} is exactly the hypothesis
\eqref{eq:norm-certificate} with $\kappa=(n+1)/(n-1)$.
Equations~\eqref{eq:certificate-holevo} and~\eqref{eq:certificate-gap}
therefore give \eqref{eq:main-holevo} and~\eqref{eq:main-gap}.
\end{proof}

\appendix
\section{Algebraic tools for reducing the word length}
\label{app:operators}

The lemmas below give the short words for Step~2 and the
factorizations for Step~3, with bounds on coefficient size and error.

\subsection{Short words in two base generators}

\begin{lemma}\label{lem:embedding}
For every integer $K\ge2$, there is an embedding
$\F_K\hookrightarrow\F_2$ whose generator images have length at most
$2\lfloor\log_2(K-1)\rfloor+1$.
Substituting these images preserves regular-representation norms of
polynomials, including those with matrix coefficients and those on
direct products of the groups.
\end{lemma}
\begin{proof}
We apply the standard covering-graph construction, keeping track of
word lengths. Take $K-1$ vertices $0,\ldots,K-2$ and the binary-tree
edges $i\xrightarrow{x}2i+1$ and $i\xrightarrow{y}2i+2$ whenever the
endpoint exists. Complete each label's partial injective map to a
permutation of the vertex set, as in
\cite[Theorem~6.1]{Stallings}. The resulting connected graph covers
the graph with one vertex and two loops labelled $x,y$.

There are $2(K-1)$ edges, of which $K-2$ belong to the binary spanning
tree. By \cite[Proposition~1A.2]{Hatcher}, the remaining $K$ edges give
a free basis: follow the tree from the root to the start of the edge,
traverse it, and return along the tree. The tree depth is at most
$\lfloor\log_2(K-1)\rfloor$, so each generator word has length at most
$2\lfloor\log_2(K-1)\rfloor+1$. The covering map embeds this free
fundamental group in $\F_2$ \cite[Proposition~1.31]{Hatcher}.

For norm preservation, let $H$ be the image subgroup. The restriction
of the regular representation of $\F_2$ to $H$ is the direct sum of
copies of the regular representation of $H$, on the invariant
subspaces $\ell^2(Hg)$. This decomposition also preserves norms with
matrix coefficients. Applying the same argument to the product
subgroup proves the assertion for direct products.
\end{proof}

\subsection{Factorization and error transfer}

We use the finite-set factorization underlying
\cite[Lemma~8.1]{BC}, together with the standard Hermitian dilation
recalled in \cite[Remark~7.4]{Pisier14}.
The following positive-matrix construction gives the bounds on the
additive constant and error transfer needed in
Section~\ref{sec:degree-reduction}.

\begin{lemma}[Factorization over a finite set]\label{lem:linear}
Let $\mathcal G$ be a group, and let $\lambda_{\mathcal G}$ be its
left regular representation on $\ell^2(\mathcal G)$, defined by
$\lambda_{\mathcal G}(g)\ket h=\ket{gh}$.
Let $S\subset\mathcal G$ be finite with $1\in S$, and let $B=|S|$. For
\[
 P(\pi)=\sum_{w\in S^{-1}S}a_w\otimes\pi(w),\qquad a_w\in\M_m(\C),
\]
there are coefficients $b_g\in\M_{2mB}(\C)$ and a scalar $\theta\ge0$,
independent of the unitary representation $\pi$, such that
$Q(\pi)=\sum_{g\in S}b_g\otimes\pi(g)$ satisfies
\begin{equation}
 \norm{P(\pi)}=\norm{Q(\pi)}^2-\theta
 \label{eq:linear-id}
\end{equation}
for every unitary representation $\pi$ on a nonzero Hilbert space,
and
\begin{equation}
 \theta\le B\norm{P(\lambda_{\mathcal G})}.
 \label{eq:theta}
\end{equation}
In particular, if $\norm{P(\lambda_{\mathcal G})}>0$ and
$\norm{Q(\pi)}\le(1+e)\norm{Q(\lambda_{\mathcal G})}$ for $0\le e\le1$, then
\[
 \norm{P(\pi)}\le(1+6Be)\norm{P(\lambda_{\mathcal G})}.
\]
\end{lemma}
\begin{proof}
Let $W=S^{-1}S$, let $\rho=\norm{P(\lambda_{\mathcal G})}$, and form the
Hermitian dilation
\[
 H(\pi)=\begin{pmatrix}0&P(\pi)\\P(\pi)^*&0\end{pmatrix}
       =\sum_{w\in W}c_w\otimes\pi(w).
\]
Thus $c_{w^{-1}}=c_w^*$, and, since $\pi$ acts on a nonzero Hilbert
space, the largest spectral value of $H(\pi)$ is $\norm{P(\pi)}$.
Compressing $H(\lambda_{\mathcal G})^*H(\lambda_{\mathcal G})$ to the
identity group vector gives
\begin{equation}
 \sum_{w\in W}c_w^*c_w\le\rho^2I.
 \label{eq:coeff-bound}
\end{equation}
Write $|c|=(c^*c)^{1/2}$.

We first construct a positive block matrix $G_0\in\M_B(\M_{2m})$
whose evaluation gives the nonconstant terms of $H$ and a controlled
constant term. For each inverse pair $\{w,w^{-1}\}$ with
$w\ne1$ and $w\ne w^{-1}$, choose one representative $w$ and
$g,h\in S$ such that $g^{-1}h=w$. Insert in block rows and columns
$g,h$ the matrix
\[
 \begin{pmatrix}|c_w^*|&c_w\\c_w^*&|c_w|\end{pmatrix}\ge0.
\]
Indeed, if $c_w=u|c_w|$ is its polar decomposition, this matrix equals
$\binom{u}{I}|c_w|(u^*\ \ I)$.
Its off-diagonal blocks give the two terms indexed by $w,w^{-1}$,
and its diagonal blocks contribute $|c_w^*|+|c_w|$ to the constant
term. If $w\ne1$ and $w=w^{-1}$, then $c_w=c_w^*$; choose
$g^{-1}h=w$ and insert half of the same matrix. The two off-diagonal
blocks then give one copy of $c_w\otimes\pi(w)$, and the diagonal
blocks contribute $|c_w|$.
Adding all these positive matrices defines $G_0\ge0$ and gives
\[
 \sum_{g,h\in S}(G_0)_{g,h}\otimes\pi(g^{-1}h)
 =H(\pi)+(D-c_1)\otimes I,
 \qquad D=\sum_{w\in W\setminus\{1\}}|c_w|.
\]

Let $\theta=\norm{D+|c_1|}$. Since
\[
 \theta I+c_1-D\ge |c_1|+c_1\ge0,
\]
we may add this matrix at the identity block to obtain
\[
 G=G_0+E_{1,1}\otimes(\theta I+c_1-D)\ge0.
\]
Here the matrix units are indexed by $S$.
Define
\[
 b_g=G^{1/2}(E_{g,1}\otimes I_{2m}),\qquad
 Q(\pi)=\sum_{g\in S}b_g\otimes\pi(g).
\]
Direct multiplication yields
\begin{align*}
 Q(\pi)^*Q(\pi)
 &=E_{1,1}\otimes
   \left(\sum_{g,h\in S}G_{g,h}\otimes\pi(g^{-1}h)\right)\\
 &=E_{1,1}\otimes\bigl(H(\pi)+\theta I\bigr).
\end{align*}
This operator is positive and its largest spectral value is
$\norm{P(\pi)}+\theta$, proving \eqref{eq:linear-id}.

To bound $\theta$, apply Cauchy--Schwarz to the vectors $|c_w|x$.
For every vector $x$, this gives the operator inequality
\[
 \left(\sum_{w\in W}|c_w|\right)^2
 \le |W|\sum_{w\in W}|c_w|^2
 =|W|\sum_{w\in W}c_w^*c_w
 \le |W|\rho^2I.
\]
Consequently,
\[
 \theta=\left\|\sum_{w\in W}|c_w|\right\|
 \le\sqrt{|W|}\,\rho\le B\rho,
\]
because $|W|=|S^{-1}S|\le B^2$. This proves \eqref{eq:theta}.
Finally, under the hypotheses for the relative-error bound, the two
norm identities and $(1+e)^2-1\le3e$ give
\[
 \begin{aligned}
 \norm{P(\pi)}-\rho
 &\le3e(\rho+\theta)\\
 &\le3(1+B)e\rho\le6Be\rho,
 \end{aligned}
\]
where the last inequality uses $B\ge1$.
\end{proof}

\section{The random-matrix estimate and parameter bounds}
\label{app:finite-threshold}

Lemma~\ref{lem:haar} makes the two-generator Bordenave--Collins
estimate explicit. Its proof follows the pair-by-pair replacement
in \cite[Lemma~9.3]{BC}, retaining the moment range of
\cite[Theorem~5.1]{BC}. Two intermediate counting bounds in the
proofs of \cite[Lemmas~5.3 and~5.9]{BC} fail on explicit paths:
one counts important exploration times, and the other counts blocks
between first and last edge visits. We replace them with corrected
bounds and a weighted enumeration of path classes, recovering the
same $N^{-1/2}$ error. The remaining subsections check the parameter
bounds used in Section~\ref{sec:approximation}.

\subsection{The two-generator Haar estimate}

\begin{lemma}[Explicit two-generator Haar estimate]
\label{lem:haar}
Let $H$ be a self-adjoint linear polynomial on $\F_2^k$, with
deterministic coefficients in $\M_m(\C)$, where $m,k\ge1$.
Let $\pi_N$ be the representation obtained by replacing the two
generators in each factor by independent Haar unitaries in $U(N)$
acting on separate tensor factors, with all $2k$ matrices independent.
Let $\lambda_2$ be the regular representation of $\F_2^k$. If
\begin{equation}
 2\le p\le2^{-2/5}N^{1/80},\qquad
 p\ \text{an even integer},
 \label{eq:haar-hypotheses}
\end{equation}
then
\[
 \E\norm{H(\pi_N)}\le\norm{H(\lambda_2)}
 (2m)^{(k+1)/p}N^{k(k+1)/(2p)}
 p^{3k(k+1)/(2p)}(1+N^{-1/2})^k.
\]
\end{lemma}
\begin{proof}
The upper bound in \cite[Lemma~9.3]{BC}, specialized to $d=2$,
has $(1+cN^{-1/2})^k$ for a universal numerical constant $c$.
The corrected one-pair estimate below gives $c=1$.
Only the self-adjoint, two-generator case is needed here.

\emph{The estimate for one pair.}
To allow iteration over tensor factors, let $\mathcal A$ be one of
the coefficient algebras $\mathcal A_j$ of the iteration below, with
its normalized faithful product trace $\tau_{\mathcal A}$.
Let $U_1,U_2$ be independent Haar unitaries in $U(N)$, and let
\[
 R_N=a_0\otimes I_N+
       \sum_{i=1}^2(a_i\otimes U_i+a_i^*\otimes U_i^*),
 \qquad a_0=a_0^*,\quad a_i\in\mathcal A,
\]
with deterministic coefficients that are finite-support polynomials:
finite sums of elementary tensors of matrices and regular shifts
$\lambda(g)$ in the free-group factors of $\mathcal A_j$.
These are the coefficients that arise in the iteration. We do not
need, and do not claim, the one-pair estimate for coefficients in an
arbitrary traced $C^*$-algebra. Let $R_{\F_2}$ be the same
polynomial evaluated at the regular generators and let
$\rho=\norm{R_{\F_2}}$.
Throughout the proof, $\norm R_p=\tau(|R|^p)^{1/p}$ uses the
appropriate normalized product trace: normalized matrix traces
on matrix factors and the canonical trace
$\tau_{\F_2}(\lambda(g))=\mathbf1_{\{g=1\}}$ on a free-group factor.
Thus the $p=2$ norm here is normalized, unlike the Hilbert--Schmidt
norm used in the main text.

The moment range is equivalent to $N\ge2^{32}p^{80}$ with
$p\ge2$; in particular, $N\ge2^{112}$. Consequently,
\[
\begin{aligned}
 256p^{20}&\le N^{1/4},&
 \alpha:=128p^{20}N^{-1/4}&\le\tfrac12,\\
 \beta:=4096p^{38}N^{-1}&\le2^{-20}p^{-42}<\tfrac12,&
 p^2N^{-1/4}&\le2^{-8}p^{-18}<1.
\end{aligned}
\]
The first inequality is precisely the hypothesis of
\cite[Theorem~5.1]{BC} for $d=2$ and moment order $p$.
The quantities $\alpha$ and $\beta$ are the two geometric ratios in
the summation below, and the last inequality controls the prefactor
of the path-weight estimate.

\emph{Paths and their weights.}
For $1\le t\le p$, let $P_t$ be the set of closed colored paths
$\gamma=(x_0,i_1,x_1,\dots,i_t,x_t)$ with
$x_s\in\{1,\dots,N\}$, $x_t=x_0$, and colors
$i_s\in\{1,2,1^*,2^*\}$ satisfying $i_{s+1}\ne i_s^*$ for $1\le s<t$.
Put $U_{i^*}=U_i^*$, let $g(\gamma)$ be the reduced word of length $t$
with $U(g(\gamma))=U_{i_1}\cdots U_{i_t}$, and let
$w(\gamma)=\E\prod_{s=1}^{t}(U_{i_s})_{x_{s-1}x_s}$. Let $v$ be the number of vertices of the colored graph $G_\gamma$ of
\cite[Definition~5.2]{BC}, and let $e_1$ be the number of its edges
traversed exactly once. The integer path defect
$\delta=t+e_1-2v\ge0$ is $2\chi$ in the notation of \cite[equation~(26)]{BC}; this
$\delta$ is local to the proof and unrelated to $\delta_K$.
The repaired argument uses the path-weight estimate
\begin{equation}
 |w(\gamma)|\le\frac{4}{N^{t/2}}\,t^{\delta}
  \Bigl(\frac{t}{N^{1/4}}\Bigr)^{e_1},
 \qquad\gamma\in P_t,
 \label{eq:path-weight}
\end{equation}
with prefactor $4$. To prove it, factor $w(\gamma)$ by independence into one mixed entry
moment for each of $U_1,U_2$. If $U_i$ and $U_i^*$ together occur
$k_i$ times, then $k_1+k_2=t$. The moment vanishes for odd $k_i$ by
Haar phase invariance, and it equals one if $k_i=0$. For even
$k_i\ge2$, \cite[Lemma~5.4]{BC}, with the explicit constants of its
proof on p.~26, bounds the moment by
\[
 \bigl(1+3k_i^{7/2}N^{-2}\bigr)\bigl(1+k_iN^{-1/4}\bigr)^{k_i}
 N^{-k_i/2}\bigl(k_iN^{-1/4}\bigr)^{e_{1,i}}k_i^{m_{4,i}/2}.
\]
Here $e_{1,i}$ is the number of simple edges with color $i$ or $i^*$,
and $m_{4,i}$ is the total multiplicity of such edges of multiplicity
at least four. The hypothesis $2k_i^{7/2}\le N^2$ of that lemma
holds. Since $k_i\le p$ and $N\ge2^{32}p^{80}$, we have
$3k_i^{7/2}N^{-2}<2^{-60}$, and the last inequality displayed above
gives $(1+k_iN^{-1/4})^{k_i}\le\exp(p^2N^{-1/4})\le\exp(2^{-26})$.
Hence the first two factors have product at most $2$.
Finally, $k_i\le t$, $e_{1,1}+e_{1,2}=e_1$, and
$m_{4,1}+m_{4,2}\le2\delta$ by \cite[equation~(27)]{BC}.
Multiplying the two bounds gives \eqref{eq:path-weight}.

To identify the terms in the moment expansion, write
\[
 R_{\F_2}^{\,p}=\sum_{|g|\le p}b_g\otimes\lambda(g),
 \qquad
 D_t=\sum_{|g|=t}\tau_{\mathcal A}(b_g)\,
                \E\operatorname{tr}_N(U(g)),
\]
where $|g|$ is reduced-word length, $U(g)$ is the word evaluated
at $U_1,U_2$, and $\operatorname{tr}_N=N^{-1}\Tr$.
The identity term is $\norm{R_{\F_2}}_p^p$, so
\[
 \E\norm{R_N}_p^p-\norm{R_{\F_2}}_p^p=\sum_{t=1}^p D_t.
\]
This is the word-length decomposition used in
\cite[equation~(21)]{BC}; odd $t$ gives zero by Haar phase invariance.

\emph{Classes of paths.}
Fix an even $t$ with $2\le t\le p$. Expanding the normalized trace in
matrix entries, as in \cite[equations~(22)--(23)]{BC}, gives
$N D_t=\sum_{\gamma\in P_t}\tau_{\mathcal A}(b_{g(\gamma)})\,w(\gamma)$.
We use the two equivalence relations of \cite[Section~5.2]{BC}.
Two paths are coarsely equivalent, $\gamma\sim\gamma'$, if they differ
by a permutation of $\{1,\dots,N\}$ and local color permutations at
the vertices, compatible at both endpoints of every edge. They are
equivalent in the refined sense, $\gamma\mathrel{\hat\sim}\gamma'$, if
in addition corresponding edges of their kernel graphs have equal
profiles \cite[Definition~5.6]{BC}. The quantities $v$, $e_1$ and
$\delta$ are constant on coarse classes, and $w$ is constant on
refined classes \cite[Lemma~5.7]{BC}. Call a refined class $\Gamma$
active if $w(\Gamma)\ne0$. Since $\tau_{\mathcal A}$ is a state,
\begin{equation}
 N|D_t|\le\sum_{\Gamma\ \mathrm{active}}|w(\Gamma)|\,
 \Bigl\|\sum_{\gamma\in\Gamma}b_{g(\gamma)}\Bigr\|.
 \label{eq:class-sum}
\end{equation}
For an active class, $v\le t/2$ by \cite[Corollary~5.5]{BC}. Hence
$s=t/2-v$ is a nonnegative integer and $\delta=e_1+2s$.

\emph{Coefficient sums.}
Let $\Gamma$ be active, and let $\hat e$ be the number of edges of its
kernel graph. We show that
\begin{equation}
 \Bigl\|\sum_{\gamma\in\Gamma}b_{g(\gamma)}\Bigr\|
 \le N^v\,2^{e_1}\,p^{9\delta+11}\rho^p .
 \label{eq:coefficient-sum}
\end{equation}
This is the conclusion of \cite[Lemma~5.9]{BC} for $d=2$, whose
exponent $18\chi+11$ equals $9\delta+11$; one counting step in its
proof must be repaired. The proof cuts each path in $\Gamma$ into
blocks. Each first or last traversal of a kernel edge is a marked
singleton block, and each maximal run of the remaining traversals is
an unmarked middle block. Let $r$ be the number of blocks and $M$ the
number of marked ones. The bound $r\le3\hat e$ of
\cite[equation~(31)]{BC} can fail. For example, the active closed path
at a single vertex with colors $(1,1,1,2,1^*,2^*,1^*,1^*)$ has
$\hat e=2$ and $r=7$. What holds is $M\le2\hat e$ and $r-M\le M-1$,
because the first block is marked and each middle block is followed by
a marked block. The proof of \cite[Lemma~5.9]{BC} bounds the left side
of \eqref{eq:coefficient-sum} by $N^v2^{e_1}\rho^p$ times a sum over
compositions $\lambda_1+\dots+\lambda_r=p$ into positive parts. Each
summand is the product of the parts $\lambda_u$ belonging to marked
blocks; middle blocks contribute no such factor. There are at most
$p^{r-1}$ compositions, and each product is at most $p^M$. The
marked-factor estimate
\[
 (r-1)+M\le3M-2\le6\hat e-2,
\]
together with $2\hat e\le3\delta+4$ \cite[equation~(29)]{BC}, bounds
the resulting exponent by $9\delta+10$. This proves
\eqref{eq:coefficient-sum}.

\emph{Weighted class enumeration.}
Explore a path step by step. Its exploration tree consists of the
$v-1$ edges along which vertices are visited for the first time. A
step is important if it traverses an edge outside this tree. Every
tree edge other than the $e_1$ simple edges is traversed at least
twice, so at most
\[
 t-2(v-1)+e_1=\delta+2
\]
steps are important. The bound $\delta+1$ asserted in
\cite[proof of Lemma~5.3, p.~25]{BC} can fail, because the final step
may traverse a tree edge for the second time. For example, the active
path with colors $(1,2,1,2^*,1^*,1^*)$ and vertex sequence
$(x_0,x,x,x,x,x,x_0)$, $x\ne x_0$, has $\delta=2$ and four important
steps, all among its first $t-1$ steps.

Label vertices by their order of first visit, and normalize colors
locally. At each vertex other than $x_0$, permute the four colors so
that the step entering it for the first time and the step leaving it
immediately afterwards receive fixed labels; at $x_0$, the first step
receives a fixed label. The mark of an important step $s$ records the
labels of $x_s$ and of the vertex at which the subsequent run of tree
steps ends, the local labels of step $s$ at $x_{s-1}$ and at $x_s$,
and the local label, at its source, of the step ending that run (a
default label if the path ends first). There are at most $64t^2$ marks. The local labels need not come from
one compatible recoloring of the path, so step $s$ is labeled at both
endpoints. Reading the steps in order recovers the vertex sequence
and all local labels, hence the coarse class, as in
\cite[proof of Lemma~5.3]{BC}. Each step that is not important either
visits a new vertex through the fixed labels or continues along the
unique non-backtracking path in the current tree toward the recorded
endpoint. Thus a coarse class is determined by the function
$c:\{1,\dots,t\}\to\{\varnothing\}\cup\{\text{marks}\}$ that records the mark of each important step and assigns
$\varnothing$ to every other step.

Fix $0<z\le1$, and give such a function the weight
$z^{|\{s:\,c(s)\ne\varnothing\}|}$. The function of a coarse class of
defect $\delta$ has weight at least $z^{\delta+2}$, and the weights of
all functions sum to $(1+64t^2z)^t$. Hence
\[
 z^{\delta+2}\cdot\#\{\text{coarse classes of defect }\delta\}
 \le(1+64t^2z)^t.
\]
This chronological encoding retains the actual positions of the
important steps instead of charging a separate factor $t$ to each of
them. We take $z=p^{-3}$. Since $t\le p$, the right side is then at
most $(1+64/p)^p$. By the profile count in
\cite[proof of Lemma~5.8]{BC}, each coarse class splits into at most
$(4t^2)^{2\hat e}$ refined classes, and $2\hat e\le3\delta+4$.
Therefore at most
\begin{equation}
 \Bigl(1+\frac{64}{p}\Bigr)^p(4p^2)^{3\delta+4}p^{3\delta+6}
 \label{eq:class-count}
\end{equation}
refined classes have defect $\delta$.

\emph{Summation.}
Consider an active class with $e_1$ simple edges and $s=t/2-v$. By
\eqref{eq:path-weight}, \eqref{eq:coefficient-sum} and $t\le p$, its
term in \eqref{eq:class-sum}, divided by $N$, is at most
\[
 \frac{4}{N^{s+1}}\,p^{\delta}
 \Bigl(\frac{2p}{N^{1/4}}\Bigr)^{e_1}p^{9\delta+11}\rho^p .
\]
Multiply by \eqref{eq:class-count} and substitute $\delta=e_1+2s$.
The active classes with given $(e_1,s)$ then contribute at most
\[
 \frac{1024\,p^{25}}{N}\Bigl(1+\frac{64}{p}\Bigr)^p
 \alpha^{e_1}\beta^{s}\rho^p
\]
to $|D_t|$. The geometric series in $e_1$ and in $s$ are each at most
$2$, so
\[
 |D_t|\le\frac{4096\,p^{25}}{N}\Bigl(1+\frac{64}{p}\Bigr)^p\rho^p .
\]
This also holds for odd $t$, since then $D_t=0$. Hence
\[
 \left|\E\norm{R_N}_p^p-\norm{R_{\F_2}}_p^p\right|
 \le\sum_{t=1}^p|D_t|
 \le\frac{4096\,p^{26}}{N}\Bigl(1+\frac{64}{p}\Bigr)^p\rho^p
 \le\frac{2^{16}p^{40}}{N}\rho^p
 \le N^{-1/2}\rho^p .
\]
The last two steps use $(1+64/p)^p\le16p^{14}$ and
$\sqrt N\ge2^{16}p^{40}$; the second is the moment range. For the first, $(1+64/p)^p$ increases
with $p$. For $2^j\le p\le2^{j+1}$ with $1\le j\le6$, it is at most
$(1+2^{5-j})^{2^{j+1}}\le2^{14j+4}\le16p^{14}$. For $p\ge128$, it is
below $e^{64}<2^{102}\le16p^{14}$.

Since $\norm{R_{\F_2}}_p\le\rho$, Jensen's inequality now gives
\[
 \E\norm{R_N}_p
 \le\rho(1+N^{-1/2})^{1/p}
 \le\rho(1+N^{-1/2}).
\]

\emph{Replacing the pairs one at a time.}
For $0\le j\le k$, let $H_j$ be $H$ with the last $j$ pairs
evaluated at their regular generators and the other pairs at their
Haar matrices. Thus $H_0=H(\pi_N)$ and $H_k=H(\lambda_2)$.
At stage $j\ge1$, condition on the first $k-j$ Haar pairs.
The coefficient algebra for the next pair, in factor $k-j+1$, is
$\mathcal A_j=\M_m(\C)\otimes\M_N(\C)^{\otimes(k-j)}
\otimes C_r^*(\F_2)^{\otimes(j-1)}$. Here $C_r^*(\F_2)$ is the norm-closed algebra generated by the
regular shifts; all tensor products act on separate Hilbert spaces.
The product trace on $\mathcal A_j$ is normalized and faithful.
The conditioned polynomial is self-adjoint, and its coefficients are
finite-support polynomials in $\mathcal A_j$ that are independent of
the next Haar pair. The one-pair estimate thus applies conditionally:
\[
 \E_{k-j+1}\norm{H_{j-1}}_p
 \le(1+N^{-1/2})\norm{H_j},
\]
where $\E_{k-j+1}$ integrates only that pair.
To repeat this step, we convert the operator norm on the right back
to a normalized $p$-norm.
For integers $r,j\ge1$ and every even integer $p\ge2$,
\cite[Lemma~9.4]{BC} states that a self-adjoint linear polynomial
$Q$ with coefficients in $\M_r(\C)$, evaluated in the regular
representation of $\F_2^j$, satisfies
$\norm Q\le(2rp^{3j})^{1/p}\norm Q_p$, where the $p$-norm uses
the normalized matrix trace and the canonical product trace.
For $1\le j<k$, the conditioned $H_j$ has this form with
$r=mN^{k-j}$, so the estimate applies throughout
\eqref{eq:haar-hypotheses} and gives
$\norm{H_j}\le(2mN^{k-j}p^{3j})^{1/p}\norm{H_j}_p$. At $j=k$, the conditional estimate already ends with
$\norm{H_k}$, so no further conversion is needed.
Starting with $\norm{H_0}\le(mN^k)^{1/p}\norm{H_0}_p$ and
integrating the pairs successively gives the following bound.
The product is empty when $k=1$, and the last inequality uses
$m\ge1$ and $p\ge2$.
\begin{align*}
 \E\norm{H_0}
 &\le(mN^k)^{1/p}
     \prod_{j=1}^{k-1}(2mN^{k-j}p^{3j})^{1/p}
     (1+N^{-1/2})^k\norm{H_k}\\
 &=\bigl[2^{k-1}m^kN^{k(k+1)/2}p^{3k(k-1)/2}\bigr]^{1/p}
      (1+N^{-1/2})^k\norm{H_k}\\
 &\le(2m)^{(k+1)/p}N^{k(k+1)/(2p)}
     p^{3k(k+1)/(2p)}(1+N^{-1/2})^k\norm{H_k}.
 \qedhere
\end{align*}
\end{proof}

\subsection{Coefficient size and error multiplication}
\label{app:reduction-costs}

We now verify the coefficient-size and error bounds in
\eqref{eq:error-size}. For this and the next subsection, assume
$K\ge2$ and $n\ge n_0(K)$.
Let $h=\ln K$, $y=1+h$, $v=\ln n$, and $u=y+v$, and retain
$q,s$ and $\ell$ from Step~3. The threshold
$n_0(K)=\lceil256y^2\rceil$ gives $n>512$, hence $v>6$ and
$u>23/3>15/2$, since $y>5/3$.

\paragraph{Summing the reduction costs.}
The definitions of $q,\ell,s$ give
\[
 q=\lceil\log_2n\rceil\le1+\frac v{\ln2}\le\frac32u,
 \qquad
 \ell\le1+\frac{2h}{\ln2}\le3y,
\]
where $1/\ln2<3/2$. Moreover, $\ln3<4/3$ and
$\ln y\le y/2$ yield
\[
 s=\lceil\log_2\ell\rceil
 \le1+\frac32\ln(3y)
 \le3+\frac34y\le u.
\]
The last inequality uses $v>6$. Thus $q+s\le5u/2$.
Also $\ln(2n)\le u$. Using $\lceil x\rceil\le x+1$ and summing
geometric series gives
\[
 \sum_{j=1}^q\left\lceil\frac n{2^j}\right\rceil
 \le n\sum_{j=1}^q2^{-j}+q\le n+q,
 \qquad
 \sum_{i=1}^s\left\lceil\frac\ell{2^i}\right\rceil
 \le\ell+s.
\]
The second sum is zero when $s=0$.

Since $1+2^{j+1}K^{\lceil n/2^j\rceil}
\le2^{j+2}K^{\lceil n/2^j\rceil}$, the support bound
\eqref{eq:coordinate-support} yields
\begin{align*}
 \sum_{j=1}^q\ln B_j
 &\le h(n+q)+\ln2\sum_{j=1}^q(j+2)\\
 &=nh+qh+\frac{q(q+5)}2\ln2\\
 &\le nh+h+\frac{hv+v^2/2}{\ln2}+\frac72v+3\ln2\\
 &\le nh+\frac34(h+v)^2+\frac72(h+v)+3\ln2\\
 &\le nh+\frac34u^2+2u.
\end{align*}
For the penultimate line, use $1/\ln2<3/2$ and expand $(h+v)^2$;
the last line follows from $h+v=u-1$ and $3\ln2<11/4$.
Likewise, \eqref{eq:word-counts}, $s\le u$, $\ell\le3y\le3u$,
and $\ln3<3/2$ give
\[
 \sum_{i=1}^s\ln C_i
 \le s\ln(2n)+(\ell+s)\ln3
 \le u^2+6u.
\]
Thus
\begin{equation}
 \sum_{j=1}^q\ln B_j\le nh+\frac34u^2+2u,\qquad
 \sum_{i=1}^s\ln C_i\le u^2+6u.
 \label{eq:summed-costs}
\end{equation}

\paragraph{Absorbing the lower-order terms into the threshold.}
For fixed $h\ge\ln2$, define
\[
 f(x)=\frac{x}{(y+\ln x)^2},\qquad
 f'(x)=\frac{y+\ln x-2}{(y+\ln x)^3}>0
 \quad(x\ge2).
\]
Let $x_0=256y^2$. Concavity of the logarithm at $2$ gives
$2\ln y\le y-2+2\ln2$, and therefore
\[
 y+\ln x_0
 =y+8\ln2+2\ln y
 \le2y+10\ln2-2<2y+5<5y.
\]
Here $\ln2<7/10$ and $y>5/3$. Hence
$f(x_0)\ge256/25$. Since $n\ge n_0(K)=\lceil x_0\rceil$,
monotonicity gives
\begin{equation}
 u^2\le\frac{25}{256}n,\qquad n\ge n_0(K).
 \label{eq:threshold-absorption}
\end{equation}

\paragraph{Total error multiplication.}
By \eqref{eq:net-norm}, the logarithm of the first factor in
\eqref{eq:error-product} satisfies
\[
 \ln\!\left[3\bigl(1+\norm{P(\lambda)}^{-1}\bigr)\right]
 \le\ln\!\left[3\left(1+\frac{K^{(n+1)/2}}{\sqrt2}\right)\right]
 \le\frac{(n+1)h}{2}+\ln(3\sqrt2)
 \le\frac{nh}{2}+u.
\]
Here $K^{(n+1)/2}/\sqrt2\ge1$, and the last step uses $n\ge2$,
$h\ge\ln2$, and $\ln(3\sqrt2)<3/2$.
Taking logarithms in \eqref{eq:error-product} and using
$q+s\le5u/2$, $\ln6<2$, and \eqref{eq:summed-costs}, we obtain
\begin{align*}
 \ln\mathfrak T_n
 &\le\frac{nh}{2}+u+(q+s)\ln6
       +\sum_{j=1}^q\ln B_j+\sum_{i=1}^s\ln C_i\\
 &\le\frac{nh}{2}+u+5u+nh+\frac74u^2+8u\\
 &=\frac32nh+\frac74u^2+14u\\
 &\le\frac32nh+\frac{217}{60}u^2
 <\frac32nh+\frac n2=\gamma_Kn.
\end{align*}
The final line uses $u>15/2$, \eqref{eq:threshold-absorption},
and $(217/60)(25/256)<1/2$.
This proves $\mathfrak T_n\le e^{\gamma_Kn}$.

\paragraph{Matrix coefficient size.}
The logarithmic contribution of the word reductions is
\begin{align*}
 \ln\!\left[\prod_{j=1}^q(2B_j)\prod_{i=1}^s(2C_i)\right]
 &=(q+s)\ln2+\sum_{j=1}^q\ln B_j+\sum_{i=1}^s\ln C_i\\
 &\le nh+\frac74u^2+\frac{21}{2}u\\
 &\le nh+\frac{63}{20}u^2
 \le nh+\frac{315}{1024}n<nh+\frac n3.
\end{align*}
We used $u>15/2$ to bound $(21/2)u\le(7/5)u^2$, and then
\eqref{eq:threshold-absorption}. Including the net size
$J_n\le(1+2n)^{K^{2n}-1}$, the initial reduction, and the final
doubling in \eqref{eq:coefficient-product} gives
\begin{align*}
 \ln(2m_n)
 &\le\ln4+\ln J_n+nh
       +\ln\!\left[\prod_{j=1}^q(2B_j)\prod_{i=1}^s(2C_i)\right]\\
 &\le\ln4+(K^{2n}-1)\ln(1+2n)+2nh+n/3.
\end{align*}
The remaining terms fit into one further copy of
$K^{2n}\ln(1+2n)$. Indeed, $\ln4\le2nh$ and $n/3\le nh/2$ imply
\[
 \ln4+2nh+n/3\le\frac92nh
 <2e\,nh\le e^{2nh}=K^{2n}
 \le K^{2n}\ln(1+2n).
\]
Here $e^{2x}\ge2ex$ for $x>0$ follows from
$e^{2x-1}\ge1+(2x-1)$, and $\ln(1+2n)>1$ for $n\ge2$.
Therefore
\[
 \ln(2m_n)
 \le(2K^{2n}-1)\ln(1+2n)
 \le2K^{2n}\ln(1+2n),
\]
which completes the two bounds in \eqref{eq:error-size}.

\subsection{Verification of the final norm error}
\label{app:haar-error}

Retain $h=\ln K$ and let $\alpha=1/80$ and $\zeta_K=2h+1/2$.
The parameters in Step~4 satisfy
\[
 b_K=40(7h+2),\qquad \alpha b_K=\gamma_K+\zeta_K,
 \qquad \varepsilon_n^{-1}=ne^{\gamma_Kn}.
\]
We will also use $b_K\le160\zeta_K$ and $\zeta_K>3/2$.
Since $N=\lceil e^{b_Kn}\rceil$,
\[
 e^{b_Kn}\le N\le2e^{b_Kn},\qquad
 \ln N\le b_Kn+\ln2\le2b_Kn.
\]
Also $N^\alpha\ge8$, so the choice
$p=2\lfloor N^\alpha/4\rfloor$ obeys
\[
 p\ge\tfrac12N^\alpha-2\ge\tfrac14N^\alpha,
 \qquad \frac1p\le4e^{-\alpha b_Kn},
 \qquad \ln p\le\alpha\ln N.
\]

Separate the three contributions in \eqref{eq:En} as
$E_n=A_n+D_n+R_n$, where
\[
 A_n=\frac{(n+1)\ln(2m_n)}p,\qquad
 D_n=\frac{n(n+1)}{2p}(\ln N+3\ln p),\qquad
 R_n=n\ln(1+N^{-1/2}).
\]
The coefficient-size bound in \eqref{eq:error-size} gives
\begin{align*}
 \frac{A_n}{\varepsilon_n}
 &\le8n(n+1)\ln(1+2n)
       e^{(2h+\gamma_K-\alpha b_K)n}\\
 &=8n(n+1)\ln(1+2n)e^{-n/2},
\end{align*}
where $2h+\gamma_K-\alpha b_K=2h-\zeta_K=-1/2$.
For the terms involving $\ln N$ and $\ln p$, we obtain
\begin{align*}
 \frac{D_n}{\varepsilon_n}
 &\le2n^2(n+1)(1+3\alpha)\ln N\,
       e^{(\gamma_K-\alpha b_K)n}\\
 &\le4b_K(1+3\alpha)n^3(n+1)e^{-\zeta_Kn}\\
 &\le664\zeta_Kn^3(n+1)e^{-\zeta_Kn},
\end{align*}
since $b_K\le160\zeta_K$ and
$4\cdot160(1+3/80)=664$.
The inequality $\ln(1+x)\le x$ gives
\[
 \frac{R_n}{\varepsilon_n}
 \le n^2e^{\gamma_Kn}N^{-1/2}
 \le n^2e^{(\gamma_K-b_K/2)n}
 \le n^2e^{-n},
\]
where $b_K/2-\gamma_K=(277/2)h+79/2>1$.

For $n\ge2$, we have $n+1\le3n/2$ and $\ln(1+2n)\le n$.
Moreover, $\zeta_K>3/2$, and $t\mapsto te^{-nt}$ decreases for
$t\ge3/2$. Thus
\[
 664\zeta_Kn^3(n+1)e^{-\zeta_Kn}
 \le664\cdot\frac32\cdot\frac32\,n^4e^{-3n/2}
 =1494n^4e^{-3n/2}.
\]
Combining the three estimates yields
\[
 \frac{E_n}{\varepsilon_n}
 \le12n^3e^{-n/2}+1494n^4e^{-3n/2}+n^2e^{-n}.
\]
Each term on the right decreases for $n\ge64$: their logarithmic
derivatives are $3/n-1/2$, $4/n-3/2$, and $2/n-1$.
At $n=64$, the bounds $e>2$ and $1494<2^{11}$ give
\begin{align*}
 12\cdot64^3e^{-32}
 +1494\cdot64^4e^{-96}+64^2e^{-64}
 &<12\cdot2^{18-32}+2^{11+24-96}+2^{12-64}\\
 &=3\cdot2^{-12}+2^{-61}+2^{-52}<\frac1{16}.
\end{align*}
Since $n\ge n_0(K)>64$, this proves $E_n\le\varepsilon_n/16$.
To compare this with the required logarithm, use
$\ln(1+x)\ge x/(1+x)$ and $0<\varepsilon_n<1$:
\[
 \ln(1+\varepsilon_n/2)
 \ge\frac{\varepsilon_n}{2+\varepsilon_n}
 >\frac{\varepsilon_n}{3}
 >\frac{\varepsilon_n}{16}
 \ge E_n.
\]
This proves \eqref{eq:En-bound}.

\smallskip
\noindent\begin{minipage}{\textwidth}
\textit{Zyphra, San Francisco, CA 94105}\\
\textit{Email:} \href{mailto:wang.jinzhao226@gmail.com}{\texttt{wang.jinzhao226@gmail.com}}
\end{minipage}

\end{document}